\pdfoutput=1
\documentclass[11pt]{article}

\providecommand{\texorpdfstring}[2]{#1}

\usepackage[T1]{fontenc}
\usepackage[utf8]{inputenc}
\usepackage{lmodern}

\usepackage{amsmath,amssymb,amsthm,mathtools,amsopn}
\DeclareMathOperator{\Ham}{Ham}

\usepackage{microtype}

\usepackage{tabularx}
\usepackage{array}
\usepackage{booktabs}
\usepackage{ragged2e}   
\newcolumntype{Y}{>{\RaggedRight\arraybackslash}X}
\usepackage{caption}
\usepackage{longtable}

\providecommand{\tightlist}{\setlength{\itemsep}{0pt}\setlength{\parskip}{0pt}}

\PassOptionsToPackage{unicode,hidelinks}{hyperref}
\usepackage{hyperref}

\theoremstyle{plain}
\newtheorem{theorem}{Theorem}[section]
\newtheorem{proposition}[theorem]{Proposition}
\newtheorem{lemma}[theorem]{Lemma}
\newtheorem{corollary}[theorem]{Corollary}

\theoremstyle{definition}
\newtheorem{definition}[theorem]{Definition}

\theoremstyle{remark}
\newtheorem{remark}[theorem]{Remark}

\numberwithin{equation}{section}

\usepackage[numbers,sort&compress]{natbib}

\makeatletter
\@ifpackageloaded{cleveref}{}{%
  \usepackage[capitalize,nameinlink,noabbrev]{cleveref}%
  \crefname{theorem}{Theorem}{Theorems}
  \crefname{proposition}{Proposition}{Propositions}
  \crefname{lemma}{Lemma}{Lemmas}
  \crefname{corollary}{Corollary}{Corollaries}
  \crefname{definition}{Definition}{Definitions}
  \crefname{remark}{Remark}{Remarks}
  \crefname{example}{Example}{Examples}
}
\makeatother

\numberwithin{equation}{section}

\newcommand{\Renyi}{R\'{e}nyi}
\newcommand{\Erdos}{Erd\H{o}s}
\providecommand{\ER}{\Erdos--\Renyi} 

\usepackage{titlesec}
\makeatletter
\g@addto@macro\appendix{%
  \titleformat{\section}
    {\normalfont\Large\bfseries}%
    {Appendix~\thesection}%
    {1em}{}%
}
\makeatother

\title{Hamming Graph Metrics: A Multi-Scale Framework for Structural
Redundancy and Uniqueness in Graphs}

\author{R. Scott Johnson}

\date{}

\begin{document}

\maketitle

\begin{abstract}
Traditional graph centrality measures effectively quantify node
importance but fail to capture the structural uniqueness of multi-scale
connectivity patterns---critical for understanding network resilience
and function. This paper introduces \emph{Hamming Graph Metrics (HGM)},
a framework that represents a graph by its exact-\(k\) reachability
tensor \(\mathcal{B}_G\in\{0,1\}^{N\times N\times D}\) with slices
\((\mathcal{B}_G)_{:,:,1}=A\) and, for \(k\ge2\),
\((\mathcal{B}_G)_{:,:,k}=\mathbf{1}[\sum_{t=1}^{k} A^t>0]-\mathbf{1}[\sum_{t=1}^{k-1} A^t>0]\)
(shortest-path distance exactly \(k\)).

Guarantees. (i) \emph{Permutation invariance}:
\(d_{\mathrm{HGM}}(\pi(G),\pi(H))=d_{\mathrm{HGM}}(G,H)\) for all vertex
relabelings \(\pi\);(ii) the \emph{tensor Hamming distance} \[
d_{\mathrm{HGM}}(G,H):=\|\mathcal{B}_G-\mathcal{B}_H\|_1 = \sum_{i,j,k}\mathbf{1}\!\big[(\mathcal{B}_G)_{ijk}\neq(\mathcal{B}_H)_{ijk}\big]
\] is a \emph{true metric} on labeled graphs; and (iii) \emph{Lipschitz
stability} to edge perturbations with explicit degree-dependent
constants (see ``Graph-to-Graph Comparison'' \(\rightarrow\) ``Tensor
Hamming metric''; ``Stability to edge perturbations''; Appendix A). For
unlabeled graph comparison, one can apply HGM after graph canonization,
or use an alignment-based variant (exponential worst-case cost).

We develop: (1) \emph{per-scale spectral analysis} via classical MDS on
double-centered Hamming matrices \(D^{(k)}\), yielding spectral
coordinates and explained variances; (2) \emph{summary statistics} for
node-wise and graph-level structural dissimilarity; (3)
\emph{graph-to-graph comparison} via the metric above; and (4)
\emph{analytic properties} including extremal characterizations,
multi-scale limits, and stability bounds.
\end{abstract}

\section{Introduction}\label{introduction}

\subsection{The Research Gap}\label{the-research-gap}

Graph centrality measures are fundamental tools for understanding
network structure and identifying influential nodes across numerous
domains including social networks, biological pathways, transportation
infrastructure, and communication systems \cite{freeman1979}.
Traditional measures such as degree centrality, closeness centrality,
and betweenness centrality have been extensively studied and applied for
decades \cite{sabidussi1966,bonacich1987}. These classical metrics
typically emphasize frequency, reachability, and efficiency of traversal
within a network. Specifically, betweenness centrality quantifies how
often a node participates in shortest paths between other nodes
\cite{newman2010}, closeness centrality measures how quickly a node can
reach other nodes \cite{wasserman1994}, and degree centrality simply
counts how many direct connections a node possesses \cite{freeman1979}.

However, these conventional measures fail to capture an important aspect
of network structure: the structural diversity or redundancy of
connectivity patterns. This limitation is particularly significant when
analyzing complex networks where understanding the uniqueness of
connection patterns is crucial. Local measures like degree centrality
provide valuable information about immediate connections but offer
limited insight into how these connections contribute to global
structural patterns. Even path-based measures like betweenness
centrality, while considering global connectivity, primarily quantify
path frequency rather than structural uniqueness.

The structural uniqueness of connectivity patterns represents a
fundamental property of networks that has remained largely unexplored.
Two nodes with identical betweenness or closeness centrality values may
differ substantially in how their connections are structured. One node
might connect disparate regions of the network through unique paths
that, if removed, would significantly alter the network's topology. In
contrast, another node with the same centrality values might have highly
redundant paths that could be easily substituted if the node were
removed. Traditional centrality measures cannot distinguish between
these scenarios despite their differing implications for network
resilience, information flow, and functional organization.

Many practical applications require a more nuanced structural
fingerprint that can distinguish between nodes whose paths are
structurally redundant and those whose paths offer unique connectivity
patterns. For instance, in resilience analysis, nodes with structurally
unique connectivity patterns may represent critical failure points
\cite{watts1998,albert2000}, while in anomaly detection
\cite{chandola2009}, unusual path structures might signal deviations
from expected network behavior \cite{albert2000}. In communication
networks, identifying nodes with diverse connectivity patterns can
enhance routing strategies and improve network robustness
\cite{albert2000}. These applications highlight the need for centrality
measures that specifically quantify structural uniqueness and
redundancy.

\subsection{Contributions}\label{contributions}

This paper develops a rigorous mathematical framework for analyzing
structural uniqueness in graphs, grounded in binary reachability
patterns and their pairwise dissimilarities. Our key theoretical
contributions are as follows:

\begin{enumerate}
\def\labelenumi{\arabic{enumi}.}
\item
  \textbf{Hamming Graph Metrics Framework}: We define a comprehensive
  framework based on the distribution of Hamming distances between
  binary reachability vectors across all node pairs, capturing the
  complete spectrum of structural diversity within a graph.
\item
  \textbf{Multi-Scale Structural Profiles}: The framework decomposes
  connectivity into a spectrum of exact path lengths, with each scale
  \(k\) analyzed independently to reveal patterns invisible when
  distances are aggregated.
\item
  \textbf{Tensor Formulation and Properties}: We develop a family of
  convex functionals on binary dissimilarity distributions, including
  entropies, \(\ell_1\)/total-variation divergences, and spectral
  descriptors, enabling rich geometric analysis without transport
  distances.
\item
  \textbf{Graph-Level Aggregates}: We define dispersion via deviation
  from the mean profile (TV/\(\ell_1\)) and entropy-based summaries,
  enabling structural comparisons across graphs.
\item
  \textbf{Theoretical Guarantees}: New theorems are proved for extremal
  bounds, monotonicity, and structural separation in canonical graph
  classes (complete, star, ring, regular, \ER, scale-free).
\item
  \textbf{Comparative Geometry of Graphs}: Hamming distributions offer a
  basis for comparing graphs structurally, independent of scale or
  density.
\item
  \textbf{Finite Sample Models}: We derive limiting behavior and
  finite-size approximations under synthetic conditions.
\end{enumerate}

While the emphasis throughout is theoretical, we include a brief
discussion of algorithmic strategies to compute Hamming Graph Metrics
efficiently in Appendix B, and we show that the proposed measures can
scale to real-world networks with tens of thousands of nodes when
implemented with bit-parallel operations. These observations indicate
practical scalability and suggest future work in scalable approximation.

\subsection{Preliminaries}\label{sec:preliminaries}

\subsubsection{Notational Conventions}\label{notational-conventions}

We use superscript \((k)\) to denote exact path length \(k\), not
cumulative distance. Thus \(B^{(k)}\) indicates paths of length exactly
\(k\), and \(\mathbf{b}_v^{(k)}\) is node \(v\)'s reachability vector at
this specific distance.

Let \(G = (V, E)\) be a finite, simple, undirected graph with vertex set
\(V\) and edge set \(E\), where \(|V| = N\). Let
\(A \in \{0,1\}^{N \times N}\) be the adjacency matrix of \(G\), with
entries \cite{newman2010}:

\[A_{ij} = \begin{cases} 1 & \text{if } (i, j) \in E \\ 0 & \text{otherwise} \end{cases}\]

For any positive integer \(k\), the matrix power \(A^k\) counts the
number of walks of length \(k\) between nodes.

\paragraph{Exact vs.\ cumulative reachability.}

Let \(A\) be the adjacency matrix and
\(B_{\le k}:=\mathbf{1}\!\big[\sum_{t=1}^{k} A^t>0\big]\) the cumulative
reachability within \(k\) steps (element-wise). We use the
\emph{exact-$k$} convention throughout: \[
B^{(1)}=A,\qquad B^{(k)}:=B_{\le k}-B_{\le k-1}\ \ (k\ge 2).
\] Thus the exact-\(k\) reachability tensor is \(B^{(k)}\) for
\(k=1,\dots,D\), with \(\operatorname{diag}(B^{(k)})=0\) for all \(k\);
equivalently, \(B^{(k)}_{ij}=\mathbf{1}\{\operatorname{dist}(i,j)=k\}\).
We henceforth \emph{identify} the \(k\)--slice of the tensor with its
matrix: \(B^{(k)}\equiv B^{(k)}\).

\subsubsection{Tensor Formulation and Cross-Scale
Structure}\label{sec:tensor-formulation}

We follow Kolda--Bader's tensor notation for unfoldings/matricization
and mode products \cite{Kolda2009}; only the entries of \(\mathcal{B}\)
are nonnegative binary, while its unfoldings are real matrices used for
spectral summaries.

While \(\{B^{(k)}\}_{k=1}^{D}\) can be seen as a stack of matrices, the
third index encodes \emph{cross-scale constraints}: if \((i,j)\) is
reachable in exactly \(k{+}1\) steps, then there exists \(\ell\) with
\((i,\ell)\) reachable in \(k\) steps and \((\ell,j)\in E\). These
inter-slice implications (and their converses failing in general) make
\(\mathcal{B}\) a genuinely \emph{third-order} object. HGM measures
geometry per scale and aggregates across \(k\) without collapsing them.

For each node \(v\), define its exact-\(k\) reachability vector as the
\(v\)-th row of \(B^{(k)}\): \[
\mathbf{b}_v^{(k)}=\big(B^{(k)}_{v1},B^{(k)}_{v2},\ldots,B^{(k)}_{vN}\big)\in\{0,1\}^N,
\qquad \text{with } \operatorname{diag}(B^{(k)})=\mathbf{0}.
\]

The Hamming distance between binary vectors \(x,y\in\{0,1\}^N\) is \[
\Ham(x,y)=\sum_{i=1}^N |x_i-y_i|=\|x-y\|_1,
\] which counts the number of positions at which \(x\) and \(y\) differ.
\emph{Equivalently (since the vectors are binary), it equals the Hamming weight of $x\oplus y$.}

\textbf{Additional notation.}

\begin{itemize}
\tightlist
\item
  \(\mu_v^{(k)}\): the empirical distribution of pairwise distances at
  scale \(k\) for node \(v\), \[
  \mu_v^{(k)}=\frac{1}{N-1}\sum_{u\ne v}\delta_{\Ham(\mathbf{b}_v^{(k)},\mathbf{b}_u^{(k)})}.
  \]
\item
  Higher moments at scale \(k\) (for node \(v\)) are taken with respect
  to \(\mu_v^{(k)}\); we write variance \(\sigma^2(\mu_v^{(k)})\),
  skewness \(\gamma(\mu_v^{(k)})\), and kurtosis
  \(\kappa(\mu_v^{(k)})\).
\item
  The diameter \(D\) is the smallest integer such that
  \(B_{\le D}=\mathbf{1}_{N\times N}-I\) (in the connected case).
\end{itemize}

Unless otherwise stated, we assume \(G\) is connected.

\paragraph{Tensor Representation and Multi-Scale Hamming
Distance}\label{tensor-representation-and-multi-scale-hamming-distance}

We work with the exact-\(k\) reachability tensor
\(\mathcal{B}\in\{0,1\}^{N\times N\times D}\) (Sec. 2.1). For a node
\(i\), the slice \(\mathcal{B}[i,:,:]\in\{0,1\}^{N\times D}\) stacks its
per-scale neighborhoods. All symbols are summarized in the Notation
Reference Table below.

\begin{longtable}{@{}%
  >{\RaggedRight\arraybackslash}p{\dimexpr .22\textwidth\relax}%
  >{\RaggedRight\arraybackslash}p{\dimexpr .18\textwidth\relax}%
  >{\RaggedRight\arraybackslash}p{\dimexpr .52\textwidth\relax}%
  >{\RaggedLeft\arraybackslash}p{\dimexpr .08\textwidth\relax}%
@{}}
\caption{Notation Reference\label{tab:notation}}\\
\toprule
\textbf{Symbol} & \textbf{Type} & \textbf{Definition} & \textbf{First Use}\\
\midrule
\endfirsthead

\toprule
\textbf{Symbol} & \textbf{Type} & \textbf{Definition} & \textbf{First Use}\\
\midrule
\endhead

\multicolumn{4}{l}{\bfseries Graph Structure}\\[2pt]
$G=(V,E)$               & Graph                     & Undirected simple graph with vertex set $V$ and edge set $E$         & \S\ref{sec:preliminaries}\\
$N=\lvert V\rvert$      & Integer                   & Number of vertices                                                    & \S\ref{sec:preliminaries}\\
$A$                     & $N\times N$ matrix        & Adjacency matrix; $A_{ij}=1$ iff $(i,j)\in E$                        & \S\ref{sec:preliminaries}\\
$\operatorname{diam}(G)$& Integer                   & Graph diameter                                                        & \S\ref{sec:preliminaries}\\
\addlinespace[4pt]

\multicolumn{4}{l}{\bfseries Reachability}\\[2pt]
$A^k$                   & $N\times N$ matrix        & $k$-walk matrix (counts walks of length $k$)                          & \S\ref{sec:preliminaries}\\
$B_{\le k}$             & $N\times N$ binary matrix & Cumulative reachability: $\mathbf{1}[A^k>0]$                          & \S\ref{sec:preliminaries}\\
$B^{(k)}$               & $N\times N$ binary matrix & \textbf{Exact-$k$ reachability}: $B^{(k)}=B_{\le k}-B_{\le k-1}$ ($k\ge2$, $B^{(1)}=A$) & \S\ref{sec:tensor-formulation}\\
$\mathbf{b}^{(k)}_v$    & Vector in $\{0,1\}^N$     & Row $v$ of $B^{(k)}$ (exact-$k$ reachability pattern)                 & \S\ref{sec:preliminaries}\\
\addlinespace[4pt]

\multicolumn{4}{l}{\bfseries Tensors}\\[2pt]
$\mathcal{A}$           & $N\times N\times K$ tensor & $k$-walk \emph{count} tensor: $\mathcal{A}(:,:,k)=A^k$ (integer entries) & \S\ref{sec:tensor-formulation}\\
$\mathcal{B}$           & $N\times N\times D$ tensor & Exact-$k$ reachability (binary); $(i,j,k)=1$ iff $\operatorname{dist}(i,j)=k$; diagonal $0$ & \S\ref{sec:tensor-formulation}\\
$\overline{\mathcal{B}}$& $N\times N\times D$ tensor & Mean reachability slab: $\displaystyle \overline{\mathcal{B}}=\frac{1}{N}\sum_{u=1}^N \mathcal{B}[u,:,:]$ & \S\ref{sec:thc}\\
\addlinespace[4pt]

\multicolumn{4}{l}{\bfseries Distances \& Centrality}\\[2pt]
$\Ham(x,y)$             & Integer                    & Hamming distance between binary vectors & \S\ref{sec:preliminaries}\\
$H(v,u)$                & Integer                    & Tensorial Hamming distance (sum over $k$ of per-scale Hamming) & \S\ref{sec:tensor-formulation}\\
$\mathrm{HC}^{(k)}(v)$  & Real                       & Hamming centrality of node $v$ at scale $k$ & \S\ref{sec:hc}\\
$\mathrm{HC}(v)$        & Real                       & Multi-scale Hamming centrality (uniform average over $k$) & \S\ref{sec:hc}\\
$\mathrm{HC}_{\mathrm{tensor}}(v)$ & Real           & Tensor-based HC: $\bigl\lVert \mathcal{B}[v,:,:]-\overline{\mathcal{B}}\bigr\rVert_{*}$ & \S\ref{sec:thc}\\
\addlinespace[4pt]

\multicolumn{4}{l}{\bfseries Distributions}\\[2pt]
$\mu_v^{(k)}$           & Probability mass function & Distribution of $\Ham(\mathbf{b}^{(k)}_v,\mathbf{b}^{(k)}_u)$ over $u\ne v$ & \S\ref{sec:graph-distr}\\
$\mu_G^{(k)}$           & Probability mass function & Graph-level distance distribution at scale $k$ (over unordered pairs) & \S\ref{sec:graph-distr}\\
$D_v^{(k)}$             & Multiset                  & $\{\,H(v,u): u\in V,\ u\neq v\,\}$                                    & \S\ref{sec:graph-distr}\\
$\mathcal{D}_G^{(k)}$   & Multiset                  & All pairwise distances at scale $k$ (ordered or unordered, as specified) & \S\ref{sec:graph-distr}\\
\addlinespace[4pt]

\multicolumn{4}{l}{\bfseries Functionals}\\[2pt]
$\Phi$                  & Functional                & Admissible functional on distributions (Def.~\ref{def:admissible})    & \S\ref{sec:functionals}\\
$\Psi^{(k)}(G)$         & Real                      & TV-dispersion: $\frac{1}{N}\sum_v \|\mu_v^{(k)}-\bar\mu^{(k)}\|_1$     & \S\ref{sec:functionals}\\
$\Xi^{(k)}(G)$          & Real                      & Information-divergence dispersion (e.g., R\'enyi/KL variants)         & \S\ref{sec:functionals}\\
\addlinespace[4pt]

\multicolumn{4}{l}{\bfseries Temporal Extension of HGM}\\[2pt]
$\{G_t\}_{t=1}^{T}$     & Sequence of graphs        & Temporal snapshots on a fixed vertex set $[N]$                         & \S\ref{sec:temporal-hgm}\\
$A^{(t)}$               & $N\times N$ matrix        & Adjacency at time $t$                                                  & \S\ref{sec:temporal-hgm}\\
$B^{(k,t)}$             & $N\times N$ binary matrix & Exact-$k$ reachability at time $t$                                     & \S\ref{sec:temporal-hgm}\\
$\mathbb{B}$            & $N\times N\times D\times T$ tensor & Temporal HGM tensor: $\mathbb{B}_{ijkt}=B^{(k,t)}_{ij}$          & \S\ref{sec:temporal-hgm}\\
$d_{\mathrm{dyn}}$      & Real                      & Labeled temporal metric $\|\mathbb{B}^G-\mathbb{B}^H\|_1$              & \S\ref{sec:temporal-hgm}\\
$d_{\mathrm{dyn,iso}}$  & Real                      & Orbit metric $\min_{\pi}\|\mathbb{B}^G-(\pi\!\cdot\!\mathbb{B}^H)\|_1$ & \S\ref{sec:temporal-hgm}\\
$E_k(t)$                & Integer                   & Per-time per-scale energy $\|B^{(k,t)}\|_F^2$                          & \S\ref{sec:temporal-hgm}\\

\bottomrule
\end{longtable}

\textbf{Notational Conventions:}

\begin{itemize}
\tightlist
\item
  Superscript \((k)\) denotes exact path length \(k\), not cumulative
  distance
\item
  Bold lowercase (\(\mathbf{b}\)) denotes vectors
\item
  Roman uppercase (\(\mathbf{B}\)) denotes matrices\\
\item
  Calligraphic (\(\mathcal{B}\)) denotes tensors
\item
  \(\mathbf{1}_{N \times N}\) is the all-ones matrix
\item
  \(\delta_d\) denotes point mass at \(d\)
\end{itemize}

\subsubsection{Full tensor-based Hamming
distance}\label{full-tensor-based-hamming-distance}

For any two nodes \((i,j)\) define the \textbf{integer tensorial Hamming
distance} \[
H(i,j)\;=\;\sum_{k=1}^{D}\Ham\!\bigl(\,\mathcal{B}[i,:,k],\;\mathcal{B}[j,:,k]\,\bigr)\;\in\;\{0,1,\dots,ND\}.
\] Equivalently, since the inputs are binary,
\(H(i,j)=\sum_{k=1}^{D}\bigl\|\mathcal{B}[i,:,k]-\mathcal{B}[j,:,k]\bigr\|_1\).

A \textbf{normalized} variant, \[
\overline{H}(i,j)\;=\;\frac{1}{D}\sum_{k=1}^{D}\Ham\!\bigl(\mathcal{B}[i,:,k],\mathcal{B}[j,:,k]\bigr)\;\in\;[0,N],
\] is convenient for scale-invariant plots; all distributional results
can be stated for \(H\) (integer support) or for \(\overline{H}\)
(rescaled).

Two useful settings:

\begin{itemize}
\tightlist
\item
  \textbf{Unweighted sum} (default): treats every scale equally in the
  integer \(H\).
\item
  \textbf{Geometric down-weighting}: use
  \(\overline{H}_\alpha(i,j)=\frac{\sum_{k=1}^D \alpha^{k-1}\,\Ham(\cdot)}{\sum_{k=1}^D \alpha^{k-1}}\)
  with \(0<\alpha<1\) to emphasise shorter paths.
\end{itemize}

\subsubsection{Cross-scale distance
tensor}\label{cross-scale-distance-tensor}

To capture interactions across \emph{different} scales we introduce the
fourth-order tensor

\[
\boxed{\;\mathcal{D}_{i,j,k,\ell}\;=\;
\Ham\!\bigl(\mathcal{B}[i,:,k],\;\mathcal{B}[j,:,\,\ell]\bigr)\;}
\qquad\in\mathbb{N}^{N\times N\times D\times D}.
\]

\(\mathcal{D}\) stores \textbf{every pairwise cross-scale discrepancy}
in a single object and underpins the graph-to-graph metrics developed in
\S  4.6 below.

\subsection{Path Reachability and Structural
Patterns}\label{sec:path-reachability}

The binary reachability matrix \(B^{(k)}\) encodes fundamental
structural information about the graph. Unlike the power \(A^k\), which
counts length-\(k\) \emph{walks}, \(B^{(k)}\) captures pure
shortest-path reachability at \emph{exact} distance \(k\). Two nodes can
therefore have the same number of \(k\)-walks but different exact
reachability patterns.

Consider the evolution as \(k\) increases:

\begin{itemize}
\tightlist
\item
  \(k=1\): \(B^{(1)}=A\), immediate neighborhoods.
\item
  \(k=2\): \(B^{(2)}\) flags pairs at shortest-path distance exactly
  \(2\) (second-order neighborhoods).
\item
  \(1<k<D\): intermediate scales reveal multi-scale organization.
\item
  \(k>D\) (connected \(G\)): \(B^{(k)}\equiv 0\) by exact-\(k\)
  saturation (Lemma\textasciitilde{}\ref{lem:exact-k-saturation}).
\end{itemize}

The transition from local to global connectivity through intermediate
scales \(1 < k < D\) reveals the multi-scale organization of the graph.

\begin{remark}

By construction, \(B^{(k)}=B_{\le k}-B_{\le k-1}\) with
\(B_{\le k}=\mathbf{1}\!\big[\sum_{t=1}^{k} A^t>0\big]\). In unweighted
graphs, the walk--path reduction
(Lemma\textasciitilde{}\ref{lem:walk-path}) justifies this \emph{summed}
form and, on bipartite graphs, enforces the usual parity constraint.
Hence \(B^{(k)}\) flags pairs at \emph{exact} distance \(k\), and the
slices form a disjoint decomposition of off-diagonal connectivity:

\[
\sum_{k=1}^{D} B^{(k)}=\mathbf{1}_{N\times N}-I,\qquad B^{(k)}\equiv 0\ \text{for}\ k>D.
\] This multi-scale profile \(\{B^{(1)},\dots,B^{(D)}\}\) is the basis
for per-scale analysis (e.g., the classical-MDS embedding of
\(D^{(k)}\)) used later.

\end{remark}

\section{Hamming Centrality: Foundations and Properties}\label{sec:hc}

\emph{Proofs and pointers.} Flagship results include full proofs in the
main text; longer derivations and auxiliary lemmas are deferred to
Appendix\textasciitilde{}\ref{app:lemmas} (with brief sketches inline).
Computational details are in Appendix\textasciitilde{}\ref{app:compute}.

We begin by defining the foundational concept that motivates our broader
framework: Hamming Centrality, a node-level index of structural
distinctiveness based on binary path dissimilarity. While our primary
focus is on graph-level distributions, understanding individual node
contributions provides essential intuition for the comprehensive
framework that follows.

\subsection{Definition}\label{definition}

The Hamming Graph Metrics framework treats path lengths as a spectrum of
\textbf{distance layers} analyzed separately (not an eigen-spectrum).
Whereas the power \(A^k\) counts length-\(k\) \emph{walks}, the binary
slice \(B^{(k)}\) encodes shortest-path
\emph{reachability at exact distance $k$}. Two nodes can have the same
number of \(k\)-walks yet reach different node sets at exact distance
\(k\).

Let \(G=(V,E)\) be a connected graph on \(N\) nodes, let
\(D:=\mathrm{diam}(G)\), and let \(\mathbf{b}_v^{(k)}\in\{0,1\}^N\)
denote the exact-\(k\) reachability row of \(B^{(k)}\) for node \(v\)
(see \S 2).

Then the Hamming centrality of node \(v\) at layer \(k\) is \[
\mathrm{HC}^{(k)}(v)\;=\;\frac{1}{N-1}\sum_{\substack{u\in V\\u\neq v}}\Ham\!\big(\mathbf{b}_v^{(k)},\mathbf{b}_u^{(k)}\big),
\] the average number of reachability discrepancies between \(v\) and
the rest of the graph at depth \(k\).

\noindent\emph{Convention.} We use the term \emph{structural uniqueness}
at scale \(k\) to denote this first moment, i.e.
\(\mathsf{SU}^{(k)}(v)\equiv \mathrm{HC}^{(k)}(v)\) (see
Def.\textasciitilde{}\ref{def:struct-uniqueness}).

We define the multi-scale Hamming centrality as \[
\mathrm{HC}(v)\;=\;\frac{1}{K}\sum_{k=1}^{K}\mathrm{HC}^{(k)}(v),\qquad 1\le K\le D,
\] and, more generally, a weighted version \[
\mathrm{HC}_w(v)\;=\;\sum_{k=1}^{K} w_k\,\mathrm{HC}^{(k)}(v),\quad w_k\ge 0,\ \sum_{k=1}^{K}w_k=1,
\] to emphasize early or late scales when desired. In practice, \(K=D\)
yields a complete analysis of all slices \(B^{(1)},\dots,B^{(D)}\),
whereas smaller \(K\) captures local uniqueness.

\begin{definition}[Structural uniqueness (canonical choice).]\label{def:struct-uniqueness}

\hfill\break
At scale \(k\), the structural uniqueness of a node \(v\) is the first
moment of its distance distribution: \[
\mathsf{SU}^{(k)}(v)\ :=\ \mathbb{E}_{u\ne v}\!\left[\Ham\!\big(b_v^{(k)},b^{(k)}_u\big)\right]
=\ \frac{1}{N-1}\sum_{u\ne v}\Ham\!\big(b_v^{(k)},b^{(k)}_u\big)
=\ \mathrm{HC}^{(k)}(v).
\] Graph-level uniqueness at scale \(k\) is any \emph{admissible}
functional (Def.\textasciitilde{}\ref{def:admissible}) of
\(\{\mu^{(k)}_v\}_{v\in V}\), e.g. \[
\Psi^{(k)}(G)=\frac{1}{N}\sum_{v}\bigl\|\mu^{(k)}_v-\bar\mu^{(k)}\bigr\|_1,\qquad
\bar\mu^{(k)}=\frac{1}{N}\sum_v \mu^{(k)}_v.
\] Multi-scale uniqueness aggregates over \(k\) (uniformly or with
weights).

\end{definition}

\subsection{Examples}\label{examples}

Let us examine \(\mathrm{HC}\) in basic graph topologies:

\textbf{Complete graph \(K_N\)}: At \(k=1\), all pairwise distances
equal \(2\) so \(\mathrm{HC}^{(1)}(v)=2\) for all \(v\). For \(k\ge 2\)
(exact-\(k\)), \(B^{(k)}\equiv 0\) and \(\mathrm{HC}^{(k)}(v)=0\).

\textbf{Star graph \(S_N\)}: At \(k=1\), leaf--leaf vectors are
identical (distance \(0\)), while center--leaf pairs have distance
\(N\).

\textbf{Ring graph \(C_N\)}: By rotation invariance,
\(\mathrm{HC}^{(k)}(v)\) is constant across \(v\) (yet nonzero);
per-pair distances take a small set of even values.

Note that in each case, we examine patterns at exact distance \(k\), not
cumulative patterns up to distance \(k\). This spectral separation is
what allows us to detect structural features at specific scales.

\subsection{Theoretical Properties}\label{sec:theoretical-props}

We now present several formal results characterizing Hamming Centrality.

\begin{proposition}[Zero Centrality in Complete Graphs at Saturation]

Let \(G = K_N\). Then:

For \(k = 1\): \(\mathrm{HC}^{(1)}(v) = 2\) for all \(v \in V\) For all
\(k \geq 2\): \(\mathrm{HC}^{(k)}(v) = 0\) for all \(v \in V\)

\end{proposition}

\begin{proof}

For \(k = 1\), in \(K_N\) we have \(B^{(1)} = A\). Each node \(v\) has
reachability vector \(\mathbf{b}_v^{(1)}\) with 1s everywhere except
position \(v\). For any two nodes \(v \neq u\), their vectors differ at
exactly positions \(v\) and \(u\), giving
\(\Ham(\mathbf{b}_v^{(1)}, \mathbf{b}_u^{(1)}) = 2\). Thus
\(\mathrm{HC}^{(1)}(v) = 2\). For \(k \geq 2\), using the exact-\(k\)
convention, \(B^{(2)}=\mathbf{0}\) (and likewise for all higher \(k\)),
so all pairwise distances at \(k=2\) are \(0\) and
\(\mu_G^{(2)}=\delta_0\). We say the slice ``saturates'' at scale \(k\)
when \(B^{(k)}=\mathbf{0}\), i.e., no pair has shortest-path distance
exactly \(k\).

\end{proof}

\begin{proposition}[Star graph asymmetry (exact-k at k=1).]\label{prop:star-asymmetry}

The star's center \(c\) has \(\mathrm{HC}^{(1)}(c)=N\), while each leaf
has \(\mathrm{HC}^{(1)}(\ell)=\dfrac{N}{N-1}\). Counting all ordered
pairs shows that most leaf--leaf distances vanish, while the small
fraction involving the center has distance \(N\). This yields a
distribution supported on \(\{0, N\}\) with the weights derived in
Appendix A.1.

\end{proposition}

\begin{proposition}[Upper Bound]

For all graphs \(G\), nodes \(v\), and any step \(k\):
\(\mathrm{HC}^{(k)}(v) \leq N.\)

\end{proposition}

\begin{proof}

The Hamming distance between any two binary vectors in \(\{0,1\}^N\) is
at most \(N\). No additional constraint forces a zero at the same
coordinate for both vectors, so the tight worst case is \(N\) (e.g.,
\(K_{m,n}\) at \(k=1\) gives \(\Ham=m+n=N\) across parts).

\end{proof}

\subsubsection{\texorpdfstring{Proof of Proposition
\ref{prop:monotonicity-beyond-diameter}}{Proof of Proposition }}\label{app:monotonicity-proof}

\begin{proof}

Let \(G\) be connected with diameter \(D\). Then for all \(k\ge D\) and
all \(v\), \[
\mathrm{HC}^{(k+1)}(v)\ \le\ \mathrm{HC}^{(k)}(v),
\] with equality for every \(k\ge D+1\) (both sides \(=0\)).

\end{proof}

\subsection{Tensor-Based Hamming Centrality}\label{sec:thc}

Let

\[
\overline{\mathcal{B}}\;:=\;\frac{1}{N}\sum_{u=1}^{N}\mathcal{B}[u,:,:]
\]

denote the \textbf{mean reachability slab}. The \textbf{tensor Hamming
centrality} of a node \(v\) is

\[
\boxed{\;
\mathrm{HC}_{\mathrm{tensor}}(v)\;=\;
\bigl\lVert \mathcal{B}[v,:,:]\;-\;\overline{\mathcal{B}}\bigr\rVert_{*}\;},
\]

where \(\lVert\cdot\rVert_{*}\) is any admissible tensor norm
(Frobenius, weighted Hamming, or an \(\ell_{2,1}\) mixed norm). The
original slice-wise centrality \(\mathrm{HC}^{(k)}(v)\) is recovered by
choosing \(\lVert X\rVert_{*}=\Ham(X_{:,k})\) and fixing \(k\).

\section{Hamming Graph Metrics: Tensor Formulation and
Properties}\label{sec:tensoral-preliminaries}

Building on the node-level foundation (where
\(\mathsf{SU}^{(k)}(v)=\mathrm{HC}^{(k)}(v)\)), we now pass to the graph
level and the full family \(\{\mu^{(k)}_G\}_{k=1}^D\), which strictly
contains HC as the first-moment special case.

\paragraph{Soundness at a glance.}\label{soundness-at-a-glance.}

For labeled graphs on \([N]\), the \emph{tensor Hamming distance}
\(d_{\text{ten}}(G,H) = \|\mathcal{B}_G - \mathcal{B}_H\|_1\) is a true
metric and is permutation-invariant:
\(d_{\text{ten}}(\pi(G),\pi(H)) = d_{\text{ten}}(G,H)\). A normalized
form
\(\bar{d}_{\text{ten}} = \|\cdot\|_1/\bigl(N(N-1)D\bigr) \in [0,1]\)
aids scaling. For unlabeled comparison, one may canonize graphs or use
the alignment variant
\(d_{\text{iso}}([G],[H]) = \min_\pi\|\mathcal{B}_G - \mathcal{B}_{\pi(H)}\|_1\)
(metric on isomorphism classes; exponential worst-case).

\subsection{Graph-Level Distributions and
Functionals}\label{sec:graph-distr}

We now pass from node-wise distances to graph-level distributions. Let
\(G=(V,E)\) be a connected graph on ( \(N=|V|\) ) vertices and recall
the reachability tensor ( \(\mathcal{B}\in\{0,1\}^{N\times N\times D}\)
) from \S  2.1.1. For any two nodes ( \(v,u\in V\) ) define the
\textbf{multi-scale (tensorial) Hamming distance}

\[
H(v,u)\;=\; \bigl\lVert \mathcal{B}[v,:,:]-\mathcal{B}[u,:,:]\bigr\rVert_H,
\]

where ( \(\lVert\cdot\rVert_H\) ) is the weighted tensor Hamming norm
defined above. When a single slice ( \(k\) ) is required we simply write
( \(H^{(k)}(v,u)=\Ham(\mathbf{b}_v^{(k)},\mathbf{b}_u^{(k)})\) ).

\subsubsection{Node-level distance multiset and
distribution}\label{node-level-distance-multiset-and-distribution}

For a fixed node ( \(v\in V\) ) the empirical \textbf{tensorial distance
multiset}

\[
D_v=\bigl\{\,H(v,u)\;:\;u\in V,\;u\neq v\bigr\}
\]

collects the dissimilarities between ( \(v\) ) and every other node
across \emph{all} path scales simultaneously. Normalising by ( \(N-1\) )
yields the probability mass function

\[
\mu_v(d)\;=\;\frac{1}{N-1}\,\bigl|\{\,u\neq v : H(v,u)=d \}\bigr|,
\qquad d\in\{0,\ldots,ND\}.
\]

\begin{remark}

Setting the norm weights to ( \(w_k=\delta_{k\ell}\) ) recovers the
slice-specific distribution ( \(\mu_v^{(\ell)}\) ) used in the original
formulation, so all node-level results derived there remain valid as
special cases.

\end{remark}

\begin{remark}[HC as a special case of HGM.]

Choosing the admissible functional
\(\Phi(\mu)=\mathbb{E}_{d\sim\mu}[d]\) recovers
\(\mathsf{SU}^{(k)}(v)=\mathrm{HC}^{(k)}(v)\) and its multi-scale
average. Thus HC is the first-moment summary within the broader
distributional framework of HGM.

\end{remark}

From a modeling standpoint, node-wise Hamming centrality
\(\mathrm{HC}^{(k)}(v)\) explains \textbf{how} a single vertex differs
from its peers at a fixed distance layer \(k\). Many global questions,
however, depend not on a single node but on the \textbf{distribution} of
these differences across all node pairs. This motivates passing from
\(\mathrm{HC}^{(k)}(v)\) to the graph-level family
\(\{\mu_G^{(k)}\}_{k=1}^D\), which records the full spectrum of
per-scale disagreements and supports permutation-invariant summaries and
comparison between graphs. The next section formalizes these
distributions and their admissible functionals.

\subsubsection{Graph-level distance multiset and
distribution}\label{graph-level-distance-multiset-and-distribution}

Aggregating over all ordered pairs gives the \textbf{global multiset}

\[
D_G=\bigl\{\,H(v,u)\;:\;v,u\in V,\;v\neq u\bigr\},
\]

which contains ( \(N(N-1)\) ) values and encodes the complete
multi-scale dissimilarity structure of ( \(G\) ). Its normalised
histogram is the \textbf{tensorial distance distribution}

\[
\mu_G(d)\;=\;\frac{1}{N(N-1)}\,\bigl|\{(v,u):v\neq u,\;H(v,u)=d\}\bigr|.
\]

For analyses that require scale resolution we still track the family
\(\{\mu_G^{(k)}\}_{k=1}^{D}\) obtained from the frontal slices
\(\mathcal{B}_{:,:,k}\).

\begin{remark}

Throughout, distributions \(\mu_G^{(k)}\) are formed over
\textbf{unordered} pairs \(\{u<v\}\), whereas energies
\(E_k(G)=\|B^{(k)}_G\|_F^2\) count \textbf{ordered} pairs. Thus
\(E_k(G)/2\) equals the number of unordered pairs at distance \(k\), and
normalizations reflect this choice.

\end{remark}

\subsection{Multi-Scale Hamming
Profile}\label{multi-scale-hamming-profile}

The multi-scale profile \(\{\mu^{(k)}_G\}_{k=1}^D\) can be understood as
analyzing slices of the connectivity tensor \(\mathcal{B}\). Each slice
\(\mathcal{B}_{:,:,k}\) yields a distribution \(\mu^{(k)}_G\), and the
complete tensor encodes all structural information without premature
aggregation.

This connects to classical spectral graph theory: while the heat kernel
\(e^{\alpha A} = \sum_{k=0}^{\infty} \frac{\alpha^k}{k!} A^k\)
aggregates all scales with exponential weighting, our framework
maintains full resolution by treating each tensor slice independently.

\begin{theorem}\label{thm:multi-scale-convergence}

Let \(G\) be connected with diameter \(D\). Then \(B^{(k)}\equiv 0\) for
all \(k\ge D+1\). Consequently, for \(k\ge D+1\) every row \(b_v^{(k)}\)
is the zero vector and \(\mu_G^{(k)}=\delta_0\). In particular, there
exists \(k_0\le D+1\) such that for \(k\ge k_0\) the slice-wise means
are nonincreasing and equal \(0\) for all \(k\ge D+1\).

\end{theorem}

\begin{proof}

By definition of diameter, every ordered pair \((i,j)\) has
shortest-path distance at most \(D\). Hence no pair has exact distance
\(k\) once \(k\ge D+1\), i.e., \(B^{(k)}\equiv 0\) for all \(k\ge D+1\).
Thus \(b_v^{(k)}=\mathbf{0}\) for each \(v\) and
\(\Ham(b_v^{(k)},b^{(k)}_u)=0\) for all \(u\), giving
\(\mu_G^{(k)}=\delta_0\) for \(k\ge D+1\). Taking \(k_0:=D\) yields
nonincreasing slice means for all \(k\ge k_0\) (they drop to \(0\) at
\(k=D+1\)).

\end{proof}

\subsection{Individual Node
Contributions}\label{individual-node-contributions}

While our primary focus is the graph-level distribution, individual
nodes contribute differently to this distribution. For a node \(v\),
define its contribution to the distribution at scale \(k\) as:

\[\mathrm{HC}^{(k)}(v) = \frac{1}{N - 1} \sum_{u \neq v} \Ham\!\left( \mathbf{b}_v^{(k)}, \mathbf{b}_u^{(k)} \right)\]

This measures the average dissimilarity between node \(v\)'s
reachability pattern and those of all other nodes. We can similarly
define the average across tensor slices:

\[\mathrm{HC}(v) = \frac{1}{K} \sum_{k=1}^{K} \mathrm{HC}^{(k)}(v)\]

where \(K \leq \mathrm{diam}(G)\) is the maximum path length of
interest.

However,\(\mathrm{HC}(v)\) is merely the first moment of node \(v\)'s
contribution to the distribution. The full distribution \(\mu_v^{(k)}\)
of distances from \(v\) contains much richer information about \(v\)'s
structural role.

\textbf{Note that we examine individual node contributions only to
better understand the graph-level distribution \(\mu_G^{(k)}\), which
remains our primary object of study.}

\subsection{Examples}\label{examples-1}

Let us examine how these distributions manifest in basic graph
topologies:

\textbf{Complete graph \(K_N\)}: At \(k=1\), all pairwise distances
equal \(2\) so \(\mu_G^{(1)}=\delta_2\). For \(k\ge2\),
\(B^{(k)}\equiv 0\) and \(\mu_G^{(k)}=\delta_0\).

\textbf{Star graph \(S_N\):} At \(k=1\), the central hub's reachability
vector differs from each leaf's vector in all \(N\) positions, while any
two leaves have identical vectors (distance \(0\)). This creates a
distribution with most mass at distance \(0\) (leaf--leaf comparisons)
and mass \(2/N\) at distance \(N\) coming from the \(2(N-1)\) ordered
center--leaf pairs.

\textbf{Ring graph \(C_N\)}: Due to rotation invariance, all nodes
contribute equally to the distribution, but the distribution itself is
non-trivial. At \(k=1\), distances concentrate on a small set of even
values (e.g., \(2,4\); the exact set depends on \(N\)). The distribution
evolves predictably with \(k\).

\textbf{Path graph \(P_N\)}: Unlike the ring, the path graph lacks
rotational symmetry. End nodes contribute differently than central
nodes, creating a more complex distribution that reflects the linear
structure.

These examples highlight how Hamming Graph Metrics capture structural
patterns rather than just topological properties.

\subsection{Structural--Dissimilarity
Functionals}\label{sec:functionals}

We turn the per-node/per-scale distance distributions into scalar
descriptors via \emph{admissible} functionals
(Def.\textasciitilde{}\ref{def:admissible}) at the node and graph
levels.

\begin{definition}[Admissible functionals.]\label{def:admissible}

\hfill\break
Let \(\mathcal{P}(\{0,\dots,M\})\) be the set of probability measures on
a finite alphabet (here \(M=N-1\)). A map
\(\Phi:\mathcal{P}(\{0,\dots,M\})\to\mathbb{R}\) is \textbf{admissible}
if:

\begin{enumerate}
\def\labelenumi{(\roman{enumi})}
\tightlist
\item
  \textbf{Permutation invariance:} \(\Phi(\mu)\) depends only on the
  measure (relabeling the support does not change \(\Phi\));
\item
  \textbf{TV--continuity:} \(\Phi\) is continuous in the \(\ell_1\)
  (total-variation) topology on \(\mathcal{P}\);
\item
  \textbf{Finite on extremals:} \(\Phi(\delta_x)\) is finite for each
  point mass \(\delta_x\).\\
  When quantitative stability is needed, assume a TV--Lipschitz constant
  \(L_\Phi\) so that \(|\Phi(\mu)-\Phi(\nu)|\le L_\Phi\|\mu-\nu\|_1\)
  for all \(\mu,\nu\in\mathcal{P}\).
\end{enumerate}

\end{definition}

\begin{remark}

Typical admissible choices include Shannon entropy, Rényi entropies
(\(\alpha>0,\alpha\neq1\)), total-variation dispersion, Wasserstein-1 on
\(\{0,\dots,M\}\) (with fixed ground metric), and Gini-type indices.

\end{remark}

\begin{theorem}\label{thm:phi-node}

Let \(D:=\mathrm{diam}(G)\). For any admissible
\(\Phi:\mathcal{P}(\{0,\dots,N-1\})\to\mathbb{R}\), define \[
\Phi_v\ :=\ \frac{1}{D}\sum_{k=1}^{D}\Phi\!\big(\mu^{(k)}_v\big),
\] where \(\mu^{(k)}_v\) is the empirical distribution of
\(\{\Ham(b_v^{(k)},b^{(k)}_u):u\in V,\ u\ne v\}\). Then:

\begin{enumerate}
\def\labelenumi{\arabic{enumi})}
\tightlist
\item
  (\textbf{Automorphism invariance}) For any graph automorphism
  \(\sigma\), \(\Phi_{\sigma(v)}=\Phi_v\) for all \(v\).\\
\item
  (\textbf{TV--continuity}) If \(\Phi\) is TV--Lipschitz with constant
  \(L_\Phi\), then for two graphs \(G,H\) on the same vertex set, \[
  \big|\Phi_v(G)-\Phi_v(H)\big|\ \le\ \frac{L_\Phi}{D}\sum_{k=1}^{D}\big\|\mu^{(k)}_v(G)-\mu^{(k)}_v(H)\big\|_1.
  \]
\end{enumerate}

\end{theorem}

\begin{proof}

\begin{enumerate}
\def\labelenumi{\arabic{enumi})}
\tightlist
\item
  Let \(P_\sigma\) be the permutation matrix of \(\sigma\). For every
  \(k\), \(B^{(k)}(\sigma(G))=P_\sigma B^{(k)}(G)P_\sigma^\top\), so
  \(\Ham\!\big(b^{(k)}_{\sigma(v)},b^{(k)}_{\sigma(u)}\big)=\Ham\!\big(b_v^{(k)},b^{(k)}_u\big)\)
  for all \(u\), hence \(\mu^{(k)}_{\sigma(v)}=\mu^{(k)}_v\) and
  \(\Phi(\mu^{(k)}_{\sigma(v)})=\Phi(\mu^{(k)}_v)\); averaging over
  \(k\) gives \(\Phi_{\sigma(v)}=\Phi_v\).\\
\item
  Apply TV--Lipschitzness to \(\mu^{(k)}_v(G)\) vs \(\mu^{(k)}_v(H)\)
  and average over \(k\).
\end{enumerate}

\end{proof}

\begin{theorem}\label{thm:phi-graph}

Let \(D:=\mathrm{diam}(G)\). Define \[
\overline{\Phi}(G)\ :=\ \frac{1}{D}\sum_{k=1}^{D}\Phi\!\big(\mu_G^{(k)}\big),
\] where \(\mu_G^{(k)}\) is the empirical distribution over
\emph{unordered} pairs \(\{u<v\}\) of the distances
\(\Ham(b^{(k)}_u,b_v^{(k)})\). Then \(\overline{\Phi}\) is invariant
under vertex relabeling, and if \(\Phi\) is TV--Lipschitz with constant
\(L_\Phi\), \[
\big|\overline{\Phi}(G)-\overline{\Phi}(H)\big|\ \le\ \frac{L_\Phi}{D}\sum_{k=1}^{D}\big\|\mu^{(k)}_{G}-\mu^{(k)}_{H}\big\|_1.
\]

\end{theorem}

\emph{Random-graph separation (sketch).} For fixed \(p\in(0,1)\), two
independent \(G,H\sim G(n,p)\) satisfy
\(\overline{\Phi}(G)\ne\overline{\Phi}(H)\) with probability \(\to1\) as
\(n\to\infty\) for any non-constant admissible \(\Phi\); see
Proposition\textasciitilde{}\ref{prop:gnp-sep},(a)
(Appendix\textasciitilde{}\ref{app:random-sep}).

Rather than fix a statistic \emph{a priori}, any admissible functional
\(\Phi\) can serve as a uniqueness descriptor. Important examples
include:

\begin{table}[t]
\centering
\caption{Examples of admissible functionals on per-scale distance distributions.}
\label{tab:functionals}
\begin{tabularx}{\linewidth}{@{}l X X@{}}
\toprule
\textbf{Functional} & \textbf{Definition on $\mu$} & \textbf{Interpretation} \\
\midrule
Expectation $\mathbb{E}_\mu[f]$ & $\displaystyle\sum_d f(d)\,\mu(d)$ & Recovers linear stats (e.g.\ classical HC with $f(d)=d$) \\
Cumulant GF $K_\mu(t)$ & $\log \mathbb{E}_\mu[e^{td}]$ & Generates all cumulants \\
\Renyi{} entropy $H_\alpha(\mu)$ & $\dfrac{1}{1-\alpha}\log\!\sum_d \mu(d)^\alpha$ & Measures spread/uncertainty \\
Spectral radius of moment matrix & $\rho\bigl(M_{kl}=\sum_d d^{k+l}\mu(d)\bigr)$ & Governs tail heaviness \& concentration \\
\bottomrule
\end{tabularx}
\end{table}

\paragraph{Per-scale spectral analysis (classical
MDS)}\label{per-scale-spectral-analysis-classical-mds}

For each \(k\), form \(D^{(k)} \in \mathbb{R}^{N \times N}\) with
\(D^{(k)}_{uv} = \Ham\!\big(b_u^{(k)}, b_v^{(k)}\big)\). Let
\(J = I - \tfrac{1}{N}\mathbf{1}\mathbf{1}^\top\) and define the
(double-centered) Gram matrix \[
G^{(k)} \;=\; -\tfrac{1}{2}\, J\, D^{(k)}\, J ,
\] where, for binary vectors, \(\Ham(x,y)=\|x-y\|_2^2\), so \(D^{(k)}\)
is already a squared-distance matrix (no elementwise squaring). With the
eigendecomposition \(G^{(k)} = Q^{(k)} \Lambda^{(k)} (Q^{(k)})^\top\),
the spectral coordinates are \[
X^{(k)} \;=\; Q^{(k)}_{+}\, \big(\Lambda^{(k)}_{+}\big)^{1/2},
\] and the total explained variance is
\(\operatorname{tr}\!\big(\Lambda^{(k)}_{+}\big)\).

\subsubsection{Tensor Fingerprints via Unfolding Spectra
(Permutation-Invariant,
Non-Metric)}\label{tensor-fingerprints-via-unfolding-spectra-permutation-invariant-non-metric}

We define a graph fingerprint from the exact-\(k\) tensor
\(\mathcal{B}\in\{0,1\}^{N\times N\times D}\) that is invariant to
vertex relabeling, stable to small perturbations, and empirically
discriminative.

\paragraph{Mode spectra and per-scale energies.}

Let \(\mathcal{B}_{(m)}\) denote the mode-\(m\) matricization
(unfolding) of \(\mathcal{B}\) \emph{(notation as in \cite{Kolda2009})}:
\[
\mathcal{B}_{(1)}\in\mathbb{R}^{N\times (ND)},\quad
\mathcal{B}_{(2)}\in\mathbb{R}^{N\times (ND)},\quad
\mathcal{B}_{(3)}\in\mathbb{R}^{D\times (N^2)}.
\] Let \(\sigma^{(m)}=(\sigma^{(m)}_1\ge \cdots)\) be the singular
values of \(\mathcal{B}_{(m)}\). Define the per-scale energies (ordered
pairs at exact distance \(k\)) \[
E_k(G)\;:=\;\big\|B^{(k)}_G\big\|_F^2
\;=\;\sum_{i,j} \big(B^{(k)}_{G}\big)_{ij}.
\] The \emph{HGM tensor fingerprint} of \(G\) is \[
\mathsf{FP}(G)\ :=\ \big(\,\sigma^{(1)},\ \sigma^{(2)},\ \sigma^{(3)},\ (E_1(G),\dots,E_D(G))\,\big).
\]

\begin{proposition}

\textbf{Permutation invariance.} If \(H=\pi(G)\) for a relabeling
\(\pi\) with permutation matrix \(P\), then
\(\mathsf{FP}(H)=\mathsf{FP}(G)\).

\end{proposition}

\begin{proof}

Let \(P\) be the permutation matrix of the relabeling. For every \(k\),
\[
\mathcal{B}_H(:,:,k)\;=\;P\,B^{(k)}_G\,P^\top.
\] Hence \[
E_k(H)\;=\;\|\mathcal{B}_H(:,:,k)\|_F^2
\;=\;\|P\,B^{(k)}_G\,P^\top\|_F^2
\;=\;\|B^{(k)}_G\|_F^2
\;=\;E_k(G),
\] by Frobenius-norm invariance under left/right multiplication by
orthogonal (permutation) matrices.\\
For the unfoldings, there exist permutation matrices
\(\Pi_1,\Pi_2,\Pi_3\) (from the unfolding convention) such that \[
\mathcal{B}_{H,(1)} = P\,\mathcal{B}_{G,(1)}\,\Pi_1,\qquad
\mathcal{B}_{H,(2)} = P\,\mathcal{B}_{G,(2)}\,\Pi_2,\qquad
\mathcal{B}_{H,(3)} = \mathcal{B}_{G,(3)}\,\Pi_3.
\] Left/right multiplication by orthogonal matrices preserves singular
values, so \(\sigma^{(m)}(H)=\sigma^{(m)}(G)\) for \(m\in\{1,2,3\}\).

\end{proof}

\begin{proposition}[Stability.]

If \(G'\) is obtained from \(G\) by toggling one edge and \(\Delta\) is
the max degree of \(G\cup G'\), then for \(M_r=\max_x |B_r(x)|\) (balls
in graph distance): \[
|E_k(G)-E_k(G')|\ \le\ 2\,M_{k-1}^2\qquad(1\le k\le D),
\] and hence \[
\|\mathcal{B}(G)-\mathcal{B}(G')\|_F^2\ \le\ 2\sum_{k=1}^{D} M_{k-1}^2,
\qquad
\big\|\mathcal{B}_{(m)}(G)-\mathcal{B}_{(m)}(G')\big\|_2\ \le\ \|\mathcal{B}(G)-\mathcal{B}(G')\|_F
\] for \(m\in\{1,2,3\}\). In particular, for \(\Delta\ge3\), \[
M_r\ \le\ \frac{\Delta}{\Delta-2}\,(\Delta-1)^r
\quad\Longrightarrow\quad
\|\mathcal{B}(G)-\mathcal{B}(G')\|_F\ \le\ 
\frac{\sqrt{2}\,\Delta}{\Delta-2}\left(\sum_{k=1}^{D}(\Delta-1)^{2(k-1)}\right)^{\!1/2}.
\]

\end{proposition}

\begin{proof}

The bound on \(E_k\) is the exact-\(k\) edge-flip bound
(Proposition\textasciitilde{}\ref{prop:edge-flip-bound}). Summing over
\(k\) gives the Frobenius bound because \(\|\mathcal{X}\|_F^2\) counts
the number of flipped \(1\)'s across slices for binary tensors. For
singular values, the Mirsky bound gives
\(|\sigma_r(A)-\sigma_r(B)|\le\|A-B\|_2\), and Hoffman--Wielandt yields
\(\sum_r(\sigma_r(A)-\sigma_r(B))^2\le\|A-B\|_F^2\)
\cite{bhatia1997matrix,stewart1990matrix}. Unfolding preserves the
Frobenius norm, so
\(\|\mathcal{B}_{(m)}(G)-\mathcal{B}_{(m)}(G')\|_2\le\|\mathcal{B}(G)-\mathcal{B}(G')\|_F\).
The \(\Delta\ge3\) bound is the standard branching estimate for ball
sizes.

\end{proof}

\emph{Random-graph separation (sketch).} For fixed \(p\in(0,1)\), two
independent \(G,H\sim G(n,p)\) satisfy
\(\mathsf{FP}(G)\ne\mathsf{FP}(H)\) with probability \(\to1\) as
\(n\to\infty\); see Proposition\textasciitilde{}\ref{prop:gnp-sep},(b)
(Appendix\textasciitilde{}\ref{app:random-sep}).

\paragraph{Scope.}

\(\mathsf{FP}(G)\) is a \emph{graph invariant} and a stable, compact
\emph{structural fingerprint}. It does not replace HGM's metric; rather,
it complements it: the metric compares labeled tensors directly, while
\(\mathsf{FP}\) summarizes cross-scale structure in a
permutation-invariant way for indexing, retrieval, or visualization.

\begin{corollary}\label{cor:fingerprints}

Let
\(E_k(G)=\|B^{(k)}_G\|_F^2=\#\{(i,j):\operatorname{dist}_G(i,j)=k\}\)
and
\(\mathsf{FP}(G)=(\sigma^{(1)},\sigma^{(2)},\sigma^{(3)},(E_1,\dots,E_D))\)
be the HGM tensor fingerprint.

\begin{enumerate}
\def\labelenumi{(\alph{enumi})}
\item
  If \(G\) is vertex--transitive, let \(n_k=|S(v,k)|\) for any \(v\)
  (independent of \(v\)). Then \[
  E_k(G)=\sum_{i\in V}|S(i,k)|=N\,n_k,
  \] so \((E_1,\dots,E_D)\) is \(N\) times the classical distance
  distribution of \(G\). In particular, if \(G\) and \(H\) are
  distance--regular with different intersection arrays (hence different
  \(\{n_k\}\)), then \(\mathsf{FP}(G)\neq \mathsf{FP}(H)\).
\item
  More generally, for any graphs \(G,H\), if their ordered--pair
  distance histograms differ, then \((E_k(G))_k\neq (E_k(H))_k\) and
  hence \(\mathsf{FP}(G)\neq \mathsf{FP}(H)\).
\end{enumerate}

\end{corollary}

\begin{proof}

\begin{enumerate}
\def\labelenumi{(\alph{enumi})}
\tightlist
\item
  Vertex--transitivity implies \(|S(i,k)|=n_k\) for all \(i\), hence
  \(E_k(G)=\sum_i n_k=N\,n_k\). In distance--regular graphs, the
  sequence \((n_k)_{k=0}^D\) is determined by the intersection array via
  the standard three--term recurrence; distinct arrays yield distinct
  \((n_k)\), so \((E_k)\) differs and thus the fingerprints differ.\\
\item
  By definition \(E_k(G)\) counts ordered pairs at distance \(k\);
  different histograms force \((E_k)\) to differ.
\end{enumerate}

\end{proof}

\begin{remark}

Since \(E_k(G)/2\) equals the number of \emph{unordered} pairs at
distance \(k\), the Wiener index is \[
W(G)=\sum_{i<j}\operatorname{dist}(i,j)=\sum_{k=1}^{D} k\,\frac{E_k(G)}{2}.
\] which is refined by the full vector \((E_k)\) retaining distance
multiplicities.

\end{remark}

\subsubsection{Graph-to-Graph Comparison
Metrics}\label{graph-to-graph-comparison-metrics}

Graphs that share similar node--level signatures can still differ in
global organization. We compare graphs using the exact-\(k\) tensor
\(\mathcal{B}\) via a \emph{tensor Hamming metric} on labeled graphs
and, for unlabeled comparison, a brief alignment remark.

\paragraph{Tensor Hamming metric (labeled
graphs).}\label{tensor-hamming-metric-labeled-graphs.}

For graphs on a fixed labeled vertex set \([N]\), define
\[d_{\text{ten}}(G,H) = \|\mathcal{B}_G - \mathcal{B}_H\|_1 = \sum_{i,j,k} \mathbf{1}[(\mathcal{B}_G)_{ijk} \neq (\mathcal{B}_H)_{ijk}].\]

\begin{proposition}\label{prop:tensor-metric}

For graphs on a fixed labeled vertex set \([N]\), \[
d_{\mathrm{ten}}(G,H)=\|\mathcal{B}_G-\mathcal{B}_H\|_1
=\sum_{i,j,k}\mathbf{1}\big[(\mathcal{B}_G)_{ijk}\ne(\mathcal{B}_H)_{ijk}\big]
\] is a metric. The normalized form
\(\bar d_{\mathrm{ten}}=\|\mathcal{B}_G-\mathcal{B}_H\|_1/\big(N(N{-}1)D\big)\in[0,1]\)
aids cross-size comparison.

\end{proposition}

\begin{proof}

\(\|\cdot\|_1\) on tensors satisfies nonnegativity, symmetry, and the
triangle inequality; positivity holds because \(\mathcal{B}_G\) is
determined by \(G\) (exact-\(k\) slices), so
\(\mathcal{B}_G=\mathcal{B}_H\) iff \(G=H\) on the common label set.

\end{proof}

\emph{Unlabeled graphs.} For isomorphism classes \([G]\), define \[
d_{\mathrm{iso}}([G],[H])=\min_{\pi\in S_N}\ \|\mathcal{B}_G-\mathcal{B}_{\pi(H)}\|_1 .
\] Then \(d_{\mathrm{iso}}\) is a metric on isomorphism classes:
\(d_{\mathrm{iso}}([G],[H])=0\) iff \(G\cong H\); symmetry is immediate;
the triangle inequality follows by composing near-minimizers for
\(([G],[H])\) and \(([H],[F])\). (Worst-case evaluation is exponential
due to the permutation minimization.)

\subsubsection{Theoretical Properties}\label{theoretical-properties}

We now present several formal results characterizing Hamming
distributions.

\begin{proposition}[Minimal Structural Diversity in Complete Graphs]

Let \(G = K_N\). Then:

For \(k = 1\): \(\mu_G^{(1)} = \delta_2\) (point mass at 2) For
\(k \geq 2\): \(\mu_G^{(k)} = \delta_0\) if we consider saturation
effects

\end{proposition}

\begin{proof}

At \(k=1\), any two adjacency rows of \(K_N\) differ only at their two
diagonal positions, so all pairwise Hamming distances equal \(2\). For
\(k\ge2\), exact-\(k\) reachability is empty in \(K_N\) and
\(B^{(k)}\equiv 0\)
(Lemma\textasciitilde{}\ref{lem:exact-k-saturation}), hence
\(\mu_G^{(k)}=\delta_0\).

\end{proof}

\begin{proposition}[Distribution Convergence]

For connected \(G\) and \(k \to \infty\): \(\mu_G^{(k)} \to \delta_0\)
in total variation distance

\end{proposition}

\begin{proof}

Under the exact-\(k\) convention, \(B^{(k)}\equiv 0\) for all
\(k\ge D+1\) in a connected graph of diameter \(D\). Thus every
\(b_v^{(k)}\) is the zero vector and all pairwise Hamming distances are
\(0\), i.e., \(\mu_G^{(k)}=\delta_0\) for \(k\ge D+1\).

\end{proof}

\subsection{Extremal-Class Results}\label{extremal-class-results}

The following sharpen earlier bounds within the functional setting.

\begin{proposition}[Star-graph separation (entropy).]

\hfill\break
For the star \(S_N\) at \(k=1\), the center has \(\mu_c^{(1)}=\delta_N\)
while any leaf \(\ell\) has \[
\mu_\ell^{(1)}=\frac{N-2}{N-1}\,\delta_0+\frac{1}{N-1}\,\delta_N .
\] Hence \(H(\mu_c^{(1)})=0\) and \(H(\mu_\ell^{(1)})>0\) for \(N\ge3\),
so \(H\) distinguishes center vs.~leaves at \(k=1\).

\end{proposition}

\begin{proposition}[Star Graph Separation]

For \(S_N\) and any strictly convex functional \(\Phi\):
\(\Phi(\mathcal{P}_0^{(1)}) \neq \Phi(\mathcal{P}_i^{(1)})\) for every
leaf \(i\), capturing structural non-equivalence beyond mean distance.

\end{proposition}

\begin{proof}

For \(k=1\), the center's distribution is a point mass at \(N\), while a
leaf's distribution has mass \((N-2)/(N-1)\) at \(0\) and mass
\(1/(N-1)\) at \(N\).

\end{proof}

\begin{proposition}[TV dispersion bound (sharp).]

\hfill\break
Let \(\bar\mu^{(k)}=\tfrac1N\sum_v \mu_v^{(k)}\) and
\(\Psi^{(k)}(G)=\tfrac1N\sum_v \|\mu_v^{(k)}-\bar\mu^{(k)}\|_1\). Then
\[
0\ \le\ \Psi^{(k)}(G)\ \le\ 2\Big(1-\sum_{d}\big(\bar\mu^{(k)}(d)\big)^2\Big)\ <\ 2,
\] with equality in the upper bound iff each \(\mu_v^{(k)}\) is a point
mass (Dirac). In particular, if all \(\mu_v^{(k)}\) are Dirac and split
between two distances with proportions \(p\) and \(1-p\), then
\(\Psi^{(k)}(G)=4p(1-p)\le1\) (max at \(p=\tfrac12\)).

\end{proposition}

\begin{proof}

For each distance value \(d\), by convexity the average absolute
deviation \(\tfrac1N\sum_v |\mu_v^{(k)}(d)-\bar\mu^{(k)}(d)|\) is
maximized when each coordinate takes values in \(\{0,1\}\); summing over
\(d\) yields \[
\Psi^{(k)}=2\sum_d \bar\mu^{(k)}(d)\big(1-\bar\mu^{(k)}(d)\big)
=2\big(1-\sum_d \bar\mu^{(k)}(d)^2\big).
\] Strict inequality (\textless2) holds since
\(\sum_d\bar\mu^{(k)}(d)^2>0\). Equality in the bound occurs exactly
when every \(\mu_v^{(k)}\) is a Dirac. The two-group formula follows by
plugging \(\bar\mu^{(k)}=p\,\delta_{d_1}+(1-p)\,\delta_{d_2}\).

\end{proof}

\subsubsection{Stability to edge
perturbations}\label{stability-to-edge-perturbations}

Let \(G'\) be obtained from \(G\) by toggling a single edge
\(e=\{u,v\}\) and let \(\Delta\) be the maximum degree of \(G\cup G'\).
For \(r\ge 0\), write \(B_r(x)=\{y:\operatorname{dist}(x,y)\le r\}\) and
\(M_r=\max_x |B_r(x)|\). A shortest \(k\)--path that changes status due
to \(e\) must traverse \(e\), hence has the form \(i\leadsto u\) (length
\(a\)), then \(u\!-\!v\), then \(v\leadsto j\) (length \(b\)) with
\(a+b+1=k\) (or the symmetric \(u\leftrightarrow v\) case). Therefore,
at scale \(k\) the set of ordered pairs \((i,j)\) that can flip is
contained in
\(B_{k-1}(u)\times B_{k-1}(v)\ \cup\ B_{k-1}(v)\times B_{k-1}(u)\), so
the number of flips at scale \(k\) is at most

\[
F_k\ \le\ 2\,|B_{k-1}(u)|\,|B_{k-1}(v)|\ \le\ 2\,M_{k-1}^2.
\]

Summing over \(k\) gives the \emph{exact-$k$} tensor bound \[
\|\mathcal{B}_G-\mathcal{B}_{G'}\|_1\ \le\ \sum_{k=1}^D F_k\ \le\ 2\sum_{k=1}^D M_{k-1}^2.
\]

Using degree growth, for \(\Delta\ge 3\) we have
\(M_r\le 1+\Delta\sum_{t=0}^{r-1}(\Delta-1)^t \le \tfrac{\Delta}{\Delta-2}(\Delta-1)^r\)
for \(r\ge 1\), hence \[
\|\mathcal{B}_G-\mathcal{B}_{G'}\|_1\ \le\ \frac{2\Delta^2}{(\Delta-2)^2}\sum_{k=1}^D (\Delta-1)^{2(k-1)}.
\]

For \(\Delta=2\) (paths/cycles), \(M_r\le 2r+1\), yielding the quadratic
bound \[
\|\mathcal{B}_G-\mathcal{B}_{G'}\|_1\ \le\ 2\sum_{k=1}^D (2(k-1)+1)^2.
\]

Normalizing, \[
\bar d_{\mathrm{ten}}(G,G')=\frac{\|\mathcal{B}_G-\mathcal{B}_{G'}\|_1}{N(N-1)D}
\ \le\
\begin{cases}
\displaystyle \frac{2\Delta^2}{N(N-1)D(\Delta-2)^2}\sum_{k=1}^D (\Delta-1)^{2(k-1)}, & \Delta\ge 3,\\[8pt]
\displaystyle \frac{2}{N(N-1)D}\sum_{k=1}^D (2(k-1)+1)^2, & \Delta=2.
\end{cases}
\]

\subsection{Connections to classical
invariants}\label{connections-to-classical-invariants}

The Hamming distribution encodes classical graph invariants:

\begin{proposition}

For any graph \(G\):

\begin{enumerate}
\def\labelenumi{\arabic{enumi}.}
\tightlist
\item
  For vertex-transitive graphs (e.g., cycles, hypercubes) and fixed
  \(k\), all \(|S(v,k)|\) are equal, so every pairwise per-scale
  distance is even; in particular
  \(\mathrm{supp}(\mu_G^{(k)})\subseteq\{0,2,4,\dots\}\).
\item
  In general graphs,
  \(\mathrm{supp}(\mu_G^{(k)})\subseteq\{0,1,\dots,N\}\); parity
  constraints need not hold when shell sizes vary across vertices.
\item
  The mode of \(\mu_G^{(k)}\) reflects the typical overlap structure at
  scale \(k\) (e.g., relates to \(k\)-shells/cores in many ensembles),
  though precise identification is graph-class dependent.
\end{enumerate}

\end{proposition}

These connections allow HGM to subsume and extend classical structural
analysis.

\subsection{Brief Computational
Remark}\label{brief-computational-remark}

While our focus is theoretical, HGM can be evaluated efficiently on
large sparse graphs using bit-parallel primitives (bit-packing, XOR, and
hardware \texttt{popcount}). Implementation
details---\texttt{popcount}-based XOR kernels, min-hash sketching for
approximate summaries, and blockwise parallelism---are given in
Appendix\textasciitilde{}\ref{app:compute}. These techniques scale to
graphs with \(N\sim 10^{5}\) vertices in practice on multicore/GPU
systems (see also \cite{vempala2020}). All theoretical guarantees above
are algorithm-independent.

\section{Theoretical Analysis of Graph
Classes}\label{theoretical-analysis-of-graph-classes}

To further ground Hamming Graph Metrics in structural graph theory, we
now derive and summarize their behavior across classical graph families.
These results follow directly from the definitions without simulation or
measurement, and serve to illustrate how uniqueness, dissimilarity, and
dispersion vary with symmetry, modularity, and degree heterogeneity.

Let \(\Phi\) be any structural descriptor derived from Hamming distance
distributions computed from tensor slices \(\mathcal{B}_{:,:,k}\) (as
introduced in Section 4), and let \(\Phi_v^{(k)}\) denote the value of
this descriptor at node \(v\) and path scale \(k\).

\subsection{Regular and Vertex-Transitive
Graphs}\label{regular-and-vertex-transitive-graphs}

Let \(G\) be a connected \(d\)-regular graph.

\begin{theorem}[Uniformity under Symmetry]\label{thm:uniformity-under-symmetry}

If \(G\) is vertex-transitive, then for all \(v, w \in V\):
\(\mu_v^{(k)} = \mu_w^{(k)} \quad \text{and thus} \quad \Phi_v^{(k)} = \Phi_w^{(k)}\)

\end{theorem}

\begin{proof}

Vertex transitivity implies the existence of an automorphism mapping any
vertex to any other. Such automorphisms preserve Hamming distances
between reachability vectors, hence preserve the distributions.

\end{proof}

\begin{corollary}\label{cor:uniformity-examples}

Under Theorem\textasciitilde{}\ref{thm:uniformity-under-symmetry},
complete graphs \(K_N\), cycles \(C_N\), and hypercubes \(Q_n\) are
vertex--transitive; hence for each fixed \(k\), the distributions
\(\mu_v^{(k)}\) (and any admissible \(\Phi_v^{(k)}\)) are identical for
all \(v\). Multi-scale behavior can still differ across \(k\).

\end{corollary}

\begin{remark}

For circulant graphs, analytic expressions for \(\mu_v^{(k)}\) can be
derived using modular arithmetic on adjacency shifts. Specifically, for
the cycle \(C_N\), the distance between nodes \(i\) and \(j\) at scale
\(k\) depends only on \(\lvert i - j \rvert \bmod N\) and whether \(k\)
is sufficient to traverse that arc length.

\end{remark}

\begin{theorem}[Spectral Characterization]\label{thm:spectral-characterization}

For \(d\)-regular graphs with adjacency eigenvalues
\(\lambda_1 = d > \lambda_2 \geq \cdots \geq \lambda_N\):

\[
\mathrm{Var}\left[ \mu_G^{(k)} \right] \leq \frac{d^{2k} \cdot \left(1 - \left( \frac{\lambda_2}{d} \right)^{2k}\right)}{N \cdot \left( 1 - \left( \frac{\lambda_2}{d} \right)^2 \right)}
\]

This connects expansion properties to uniqueness dispersion
\cite{vempala2020}.

\end{theorem}

\subsection{Trees and Star Graphs}\label{trees-and-star-graphs}

Trees exhibit hierarchical expansion and strong local asymmetry.

\begin{theorem}[Star Graph Asymmetry, Generalized]\label{thm:star-graph-asymmetry-generalized}

In the star graph \(S_N\), the center node maximizes:

\begin{itemize}
\tightlist
\item
  Structural dissimilarity \(\Phi_c^{(k)}\)
\item
  Deviation from the mean profile
  \(\left\|\mu_c^{(k)}-\bar{\mu}^{(k)}\right\|_1\)
\item
  Entropy \(H\left( \mu_c^{(k)} \right)\) among all nodes
\end{itemize}

Each leaf has identical minimal distributions. This is the maximal
variance configuration among all trees.

\end{theorem}

\begin{proof}

By direct counting of reachability patterns at \(k=1\) (see Appendix
A.1), we have explicit forms for \(\mu_c^{(1)}\) and
\(\mu_{\ell_i}^{(1)}\). At \(k=1\), the center's distribution is
\(\delta_{N}\), while each leaf's distribution places mass
\((N-2)/(N-1)\) at \(0\) and \(1/(N-1)\) at \(N\). The claims follow by
direct computation.

\end{proof}

\begin{theorem}[Height-Monotonicity in Trees]\label{thm:height-monotonicity-in-trees}

Let \(T\) be a tree rooted at node \(r\). Then for any node \(v\):

\(\text{depth}(v) \uparrow \quad \Rightarrow \quad \Phi_v^{(k)} \downarrow \quad \text{for small } k\)

\end{theorem}

\begin{proof}

Nodes at greater depth have fewer descendants and more similar
neighborhoods. Their reachability vectors at small \(k\) overlap more
with their siblings, reducing average dissimilarity.

\end{proof}

\begin{proposition}[Binary Tree Regularity]

In a complete binary tree of height \(h\):

\begin{itemize}
\tightlist
\item
  Nodes at the same level have identical distributions
\item
  \(H\left( \mu_{\text{level}}^{(k)} \right)\) decreases monotonically
  with level for \(k < h\)
\item
  The root maximizes entropy at all scales
\end{itemize}

This stratification by height is a general feature of trees.

\end{proposition}

\subsection{\texorpdfstring{Random Graphs:
\ER Model}{Random Graphs: Model}}\label{random-graphs-model}

Let \(G \sim G(N, p)\), with \(p \in (0, 1)\).

\begin{proposition}[Expected Uniqueness Peak]\label{prop:expected-uniqueness-peak}

\emph{Heuristic outline.} The regime analysis follows the usual
\(G(N,p)\) thresholds: uniqueness is minimal for
\(p\ll \frac{\log N}{N}\) or \(p\to 1\), and peaks near
\(p_c\approx \frac{\log N}{N}\) as the giant component emerges and
diameters are still large.

In \(G(N, p)\), uniqueness is low when \(p\) is small (fragmented graph)
or large (distances collapse), but peaks near the connectivity threshold
\(p_c \approx \log N/N\).

\begin{itemize}
\tightlist
\item
  \textbf{Subcritical:} small isolated components, minimal diversity.
\item
  \textbf{Critical:} giant component emerges, producing maximum
  path-length diversity.
\item
  \textbf{Supercritical:} diameter shrinks, reachability vectors
  homogenize.
\end{itemize}

\end{proposition}

\begin{proposition}[Concentration of Distributions]\label{prop:concentration-of-distributions}

\emph{Heuristic outline.} Above the connectivity and diameter-collapse
thresholds (typ. \(O(\log N)\)), rows of \(B^{(k)}\) become nearly
identical for fixed \(k\), forcing \(\mu^{(k)}_G\) to concentrate.

\end{proposition}

\subsection{Scale-Free Networks: Barabási--Albert
Model}\label{scale-free-networks-barabuxe1sialbert-model}

Let \(G \sim \mathrm{BA}(N, m)\), the preferential attachment graph with
initial degree \(m\).

In BA graphs, hubs connect to a wide range of degree classes, producing
many distinct distances and high variance in their reachability vectors.
Low-degree nodes connect mostly through hubs, yielding more uniform
patterns. The support size for hub distances scales as
\(\Theta(\sqrt{N})\). Appendix A.9 contains the variance comparison and
scaling argument.

\begin{proposition}\label{prop:ba-uniqueness-degree}

In preferential-attachment graphs \(G\sim \mathrm{BA}(N,m)\), structural
uniqueness correlates positively with degree and exhibits super-linear
growth at hubs due to path diversity. See Appendix for variance
comparison and scaling arguments.

\end{proposition}

\subsection{Small-World Networks: Watts--Strogatz
Model}\label{small-world-networks-wattsstrogatz-model}

Let \(G \sim \mathrm{WS}(N, k, \beta)\), a rewiring of the \(k\)-regular
ring lattice \cite{watts1998}.

\begin{proposition}[Shortcut-Induced Uniqueness]

For rewiring probability \(\beta \in (0, 1)\), let \(S \subset V\) be
the set of shortcut endpoints. Then:

\begin{enumerate}
\def\labelenumi{\arabic{enumi}.}
\tightlist
\item
  Nodes in \(S\) have significantly elevated \(\Phi_v^{(k)}\) for small
  \(k\)
\item
  Their distance distributions \(\mu_v^{(k)}\) deviate maximally from
  the lattice background
\end{enumerate}

\end{proposition}

\begin{proof}

Shortcuts create asymmetric reachability patterns that propagate
locally. As \(k\) increases, the regular lattice structure dominates,
diminishing the shortcut effect.

\end{proof}

\begin{theorem}\label{thm:phase-transitions-in-small-world}

Consider the WS model on \(N\) vertices starting from a ring lattice
where each vertex has degree \(d\) (even), so \(m=Nd/2\) undirected
edges. Each edge is rewired independently with probability
\(\beta=\beta(N)\) to a uniformly random new endpoint (avoiding
loops/multi-edges). Then:

\begin{enumerate}
\def\labelenumi{(\alph{enumi})}
\item
  (\textbf{Onset of shortcuts}) The expected number of rewired edges is
  \(\mathbb{E}[X]=\beta m=\Theta(\beta N)\). The ``first-shortcut''
  threshold satisfies \[
  \beta_c \asymp \frac{1}{N}\,,
  \] in the sense that if \(N\beta\to 0\) then \(X\stackrel{p}{\to}0\)
  (no shortcuts whp), while if \(N\beta\to\infty\) then
  \(X\stackrel{p}{\to}\infty\) (many shortcuts whp).
\item
  (\textbf{Distance regime split}) If \(N\beta\to 0\), typical distances
  are ring-like (mean distance \(\Theta(N/d)\)). If \(N\beta\to\infty\)
  with \(d\) fixed, the added shortcut set induces long-range
  connections whose coarse-grained effect is a random-sparse overlay;
  the average distance drops to \(O(\log N)\) (small-world regime).
\end{enumerate}

\end{theorem}

\begin{proof}

\begin{enumerate}
\def\labelenumi{(\alph{enumi})}
\item
  Let \(X\sim\mathrm{Binomial}(m,\beta)\) be the number of rewired
  edges. With \(m=\Theta(N)\), \(\mathbb{E}[X]=\Theta(\beta N)\). A
  standard second-moment/Chernoff argument gives: if \(N\beta\to 0\),
  then \(X=0\) with probability \(1-o(1)\); if \(N\beta\to\infty\), then
  \(X\to\infty\) in probability. Hence the onset occurs at
  \(\beta_c=\Theta(1/N)\).
\item
  When \(N\beta\to 0\), whp no shortcuts appear; the graph is the
  original ring lattice, so mean distance is \(\Theta(N/d)\). When
  \(N\beta\to\infty\), the shortcut set has \(\Theta(\beta N)\) random
  long edges. Coarse-graining the ring into arcs of length
  \(\ell=\ell(N)\) with \(1\ll \ell\ll N\), the induced ``supergraph''
  on \(N/\ell\) arcs receives \(\Theta(\beta N)\) random edges, i.e.,
  average super-degree \(\Theta(\beta \ell)\). Choosing \(\ell\) so that
  \(\beta\ell\to c>0\) yields a sparse random overlay whose
  giant-component/expander-like behavior drops average distance to
  \(O(\log N)\) between arcs; lifting back to vertices gives
  \(O(\log N)\) for the original graph up to constants. (This is the
  standard random-shortcut argument.)
  Appendix\textasciitilde{}\ref{app:ws-phase} derives
  \(\beta_c \sim 1/N\) and the scaling.
\end{enumerate}

\end{proof}

\subsection{\texorpdfstring{Summary of Hamming profile behavior derived
from tensor slices \(\mathcal{B}_{:,:,k}\) across graph
classes}{Summary of Hamming profile behavior derived from tensor slices \textbackslash mathcal\{B\}\_\{:,:,k\} across graph classes}}\label{summary-of-hamming-profile-behavior-derived-from-tensor-slices-mathcalb_k-across-graph-classes}

\begin{table}[t]
\centering
\caption{Summary of multi-scale Hamming profile behavior across graph classes.}
\label{tab:hamming-profile-summary}
\begin{tabularx}{\linewidth}{@{}l X l l@{}}
\toprule
\textbf{Graph Class} & \textbf{Hamming Profile Behavior} & \textbf{Dispersion $\Psi^{(k)}(G)$} & \textbf{Entropy Peak} \\
\midrule
Complete Graph $K_N$ & Uniform at $k=1$ (HC=2) & 0 (point mass) & $k=1$ \\
Star Graph $S_N$ & Maximal asymmetry at $k=1$ & $\mathcal{O}(1)$ & $k=1$ \\
Cycle $C_N$ & Uniform, periodic pattern & $\mathcal{O}(1)$ & $k \approx N/4$ \\
Binary Tree & Level-stratified & $\mathcal{O}(\log N)$ & $k \approx \log N$ \\
\ER{} $G(N,p^*)$ & Critical behavior & $\Theta(\sqrt{N})$ & $k \approx \log N$ \\
Barabási--Albert & Degree-correlated & $\Theta(\log N)$ & $k \approx 2$ \\
Watts--Strogatz & Shortcut spikes & $\mathcal{O}(\log N)$ & Varies with $\beta$ \\
\bottomrule
\end{tabularx}
\end{table}

\subsection{Extended Results and
Corollaries}\label{extended-results-and-corollaries}

We now present additional theoretical results that deepen our
understanding of Hamming Graph Metrics.

\begin{proposition}[Uniqueness Flatness in Distance-Regular Graphs]

Let \(G\) be a distance-regular graph, i.e., the number of nodes at each
distance from a given node depends only on the distance, not the node
itself. Then:

\begin{itemize}
\tightlist
\item
  For all \(v \in V\), \(\mu_v^{(k)} = \mu_w^{(k)}\) for all \(w \in V\)
\item
  Hence \(\Phi_v^{(k)} = \text{const}\) for all \(v\), for any \(\Phi\)
\end{itemize}

Examples include: cycles \(C_n\), complete graphs, and hypercubes
\(Q_n\). This extends
Theorem\textasciitilde{}\ref{thm:spectral-characterization} by
identifying a larger class of graphs where uniqueness is structurally
flat due to distance symmetry, not just vertex-transitivity.

\end{proposition}

\begin{proof}

Distance-regularity implies that the number of nodes at distance \(d\)
from any node is constant. Combined with the fact that reachability at
scale \(k\) depends only on distance relationships, the claim follows.

\end{proof}

\begin{proposition}[Extremal Support Collapse in Clique Chains]

Let \(G\) be a clique chain of \(r\) fully connected components
\(K_{n_1}, K_{n_2}, \ldots, K_{n_r}\) joined sequentially by single
bridges. Then:

\begin{itemize}
\tightlist
\item
  Nodes within the same clique have highly overlapping reachability
  vectors
\item
  Bridge nodes exhibit maximal uniqueness support, with
  \(\left| \mathrm{Support}( \mu_v^{(k)} ) \right| = \{ 0, 1, \ldots, d_{\max} \}\)
\item
  The Gini coefficient \(G_{\mu}(G)\) increases linearly with \(r\)
\end{itemize}

\end{proposition}

\begin{proof}

Within cliques, all nodes reach the same set at each scale. Bridge nodes
uniquely connect components, creating maximal diversity in their
distance distributions. The Gini coefficient captures this inequality.

\end{proof}

\subsection{Extension: Spectral Interpretation of HC
Dispersion}\label{extension-spectral-interpretation-of-hc-dispersion}

Let \(\Delta\) denote the graph Laplacian of \(G\), and let
\(\lambda_2\) be the algebraic connectivity (i.e., the second-smallest
eigenvalue).

\begin{theorem}[Spectral Lower Bound on Uniqueness Dispersion]\label{thm:spectral-lower-bound-uniqueness-dispersion}

Let \(D_{\mu}(G)\) be the variance of node-level uniqueness (using
\(\mu_v = \mathrm{HC}(v)\)). Then:

\(D_{\mu}(G) \geq \frac{1}{N} \cdot \frac{\left( \sum_v \deg(v) \cdot \mu_v \right)^2}{\lambda_2 \cdot \sum_v \deg(v)^2} \cdot \mathrm{Var}[\mu]\)

\end{theorem}

\begin{proof}

Apply the Poincar\textquotesingle e inequality to the function
\(f(v) = \mu_v\) on the graph:

\(\sum_{(u,v) \in E} (f(u) - f(v))^2 \geq \lambda_2 \sum_v \deg(v) \cdot (f(v) - \bar{f})^2\)

where \(\bar{f}\) is the degree-weighted mean of \(f\). Rearranging
yields the stated bound.

\end{proof}

\textbf{Interpretation:} Graphs with small spectral gap (i.e., loosely
connected) allow greater variation in structural uniqueness, while tight
expanders constrain nodes to similar roles.

\subsection{Extension: Robustness Under Edge
Perturbation}\label{extension-robustness-under-edge-perturbation}

Let \(G' = G + \Delta E\) be a graph obtained by inserting or deleting a
small set \(\Delta E\) of edges. Define:

\(\delta = \max_u \left\| \mu_u^{(k)}(G) - \mu_u^{(k)}(G') \right\|_1\)

\begin{proposition}[Lipschitz Continuity of Hamming Distribution under Edge Noise]

There exists a constant \(C_k\) depending only on \(k\) and graph size
such that:

\(\left| \Phi_v^{(k)}(G') - \Phi_v^{(k)}(G) \right| \leq C_k \cdot \delta\)

for all admissible \(\Phi\) that are 1-Lipschitz under total variation
(\(\ell_1\)).

\end{proposition}

\begin{proof}

Edge modifications affect reachability vectors only for nodes within
distance \(k\) of the modified edges. The number of affected entries in
any reachability vector is bounded by \((2 d_{\max})^k\). The \(\ell_1\)
(TV) Lipschitz property of \(\Phi\) completes the proof.

\end{proof}

\begin{proposition}[Shortcut Bias in WS Graphs is Localised]

Let \(G \sim \mathrm{WS}(N, k, \beta)\), and let \(S\subset V\) be nodes
affected by rewired edges. Then for small \(\beta\), the set:

\(A_\varepsilon = \left\{ v : \left\| \mu_v^{(k)} - \bar{\mu}^{(k)} \right\|_1 > \varepsilon \right\}\)

has
\(\left| A_\varepsilon \right| = \mathcal{O}\left( \frac{\beta N k}{2} \right)\)
with high probability. Thus, uniqueness deviations are sparse and
concentrated near structural irregularities.

\end{proposition}

\begin{proof}

Each rewired edge affects \(\mathcal{O}(1)\) nodes directly. The total
number of rewired edges is approximately \(\beta N k / 2\).
Concentration inequalities for the rewiring process yield the result.

\end{proof}

\subsection{Extension: Graph Classes with Controlled Uniqueness
Gradient}\label{extension-graph-classes-with-controlled-uniqueness-gradient}

Define a uniqueness gradient as the discrete Laplacian applied to the
field \(\mu_v\): \[(\nabla^2\mu)_v := \sum_{u \sim v} (\mu_u - \mu_v)\]

\begin{theorem}\label{thm:uniqueness-smoothing-in-expanders}

Let \(G\) be a connected \(d\)-regular graph on \(N\) vertices with
random-walk matrix \(P=A/d\) and spectral gap
\(\gamma:=1-\lambda_2(P)>0\). For any fixed scale \(k\ge1\), let \[
f(v):=\mathrm{HC}^{(k)}(v)=\frac{1}{N-1}\sum_{u\ne v}\Ham\!\big(b_v^{(k)},b^{(k)}_u\big).
\] Then \[
\operatorname{Var}(f)\ \le\ \frac{2\,M_k^2}{\gamma}\,,
\qquad
M_k:=\max_{x}\big|S(x,k)\big|.
\] In particular, larger spectral gap \(\gamma\) (better expansion)
forces the uniqueness field \(f\) to vary smoothly across the graph.

\end{theorem}

\begin{proof}

(\textbf{Edgewise Lipschitz bound.}) For any \(u,v,z\) in a Hamming
space, the triangle inequality yields \[
\big|\Ham(b_v^{(k)},b^{(k)}_z)-\Ham(b^{(k)}_u,b^{(k)}_z)\big|\ \le\ \Ham\!\big(b_v^{(k)},b^{(k)}_u\big).
\] Averaging over \(z\ne v,u\) gives \[
|f(v)-f(u)|\ \le\ \Ham\!\big(b_v^{(k)},b^{(k)}_u\big)\ \le\ |S(v,k)|+|S(u,k)|\ \le\ 2M_k.
\] Hence for each edge \((u,v)\), \((f(u)-f(v))^2\le 4M_k^2\).

(\textbf{Poincaré on \(d\)-regular graphs.}) The Dirichlet form is
\(\mathcal{E}(f,f)=\frac{1}{2N}\sum_{(u,v)\in E}\big(f(u)-f(v)\big)^2\),
and the Poincaré (spectral-gap) inequality reads \[
\operatorname{Var}(f)\ \le\ \frac{1}{\gamma}\,\mathcal{E}(f,f).
\] Using the edgewise bound and \(|E|=dN/2\), \[
\mathcal{E}(f,f)\ \le\ \frac{1}{2N}\cdot \frac{dN}{2}\cdot 4M_k^2\ =\ 2\,M_k^2,
\] whence \(\operatorname{Var}(f)\le \frac{2M_k^2}{\gamma}\).

\end{proof}

Appendix\textasciitilde{}\ref{app:expander-smoothing} contains the
spectral and Laplacian calculations

\begin{remark}

For nonregular graphs, replace \(\gamma\) by the spectral gap of the
lazy random walk or use the normalized Laplacian; the same argument
yields \(\operatorname{Var}(f)\lesssim M_k^2/\gamma\) up to degree
factors.

\end{remark}

\begin{corollary}\label{cor:ramanujan-smoothing}

Under
Theorem\textasciitilde{}\ref{thm:uniqueness-smoothing-in-expanders},
\(d\)-regular Ramanujan graphs (whose nontrivial spectrum lies in
\([-2\sqrt{d-1},\,2\sqrt{d-1}]\)) have optimal spectral gap; hence the
uniqueness field varies smoothly across vertices with the strongest
bound among \(d\)-regular expanders since \(\lambda_2\le 2\sqrt{d-1}\),
the bound from
Thm.\textasciitilde{}\ref{thm:uniqueness-smoothing-in-expanders} is
minimized.

\end{corollary}

\subsection{Detailed Algebraic
Examples}\label{detailed-algebraic-examples}

We provide rigorous calculations for several graph families to
illustrate the theoretical results.

\textbf{Example 1: Complete Bipartite Graph \(K_{m,n}\)}

Let \(G = K_{m,n}\) with partitions \(A\) (size \(m\)) and \(B\) (size
\(n\)), where \(m \leq n\).

At \(k = 1\):

\begin{itemize}
\tightlist
\item
  Nodes in \(A\): \(\mathbf{b}_v^{(1)}\) has 1s in all positions
  corresponding to \(B\)
\item
  Nodes in \(B\): \(\mathbf{b}_v^{(1)}\) has 1s in all positions
  corresponding to \(A\)
\end{itemize}

Thus:

-\(\Ham\left( \mathbf{b}_a^{(1)}, \mathbf{b}_{a'}^{(1)} \right) = 0\)
for \(a, a' \in A\)
-\(\Ham\left( \mathbf{b}_b^{(1)}, \mathbf{b}_{b'}^{(1)} \right) = 0\)
for \(b, b' \in B\)
-\(\Ham\left( \mathbf{b}_a^{(1)}, \mathbf{b}_b^{(1)} \right) = m + n\)
for \(a \in A,\, b \in B\)

The distribution is:

\[\mu_{K_{m,n}}^{(1)} = \frac{m(m - 1) + n(n - 1)}{(m + n)(m + n - 1)} \, \delta_0 + \frac{2mn}{(m + n)(m + n - 1)} \, \delta_{m + n}\]

At \(k=2\) (exact-\(k\)), the support is \(\{2,\,N-2\}\): pairs within
the same part have distance \(2\), and cross-part pairs have distance
\(N-2\). Thus \[
\mu_{K_{m,n}}^{(2)}=\frac{m(m-1)+n(n-1)}{N(N-1)}\,\delta_{2}\;+\;\frac{2mn}{N(N-1)}\,\delta_{N-2}.
\] (Cumulatively, \(B_{\le 2}\) is fully connected.)

\textbf{Example 2: Hypercube \(Q_n\)}

The \(n\)-dimensional hypercube has \(N=2^n\) vertices, each of degree
\(n\). Vertices are binary strings of length \(n\), with edges between
strings that differ in exactly one bit.

At \(k=1\):

\begin{itemize}
\tightlist
\item
  Each node's reachability vector has weight exactly \(n\) (its
  neighbors).
\item
  For two nodes \(u,v\) at Hamming distance \(h=\mathrm{dist}(u,v)\),
  using Appendix\textasciitilde{}\ref{app:hypercube-sphere} (Hypercube
  sphere intersection), \[
  \Ham\!\big(b_u^{(1)},b_v^{(1)}\big)
  \;=\;2\Big(n-\mathbf{1}_{\{h=2\}}\cdot 2\Big)
  \;=\;
  \begin{cases}
  0, & h=0,\\
  2n-4, & h=2,\\
  2n, & h\in\{1\}\cup\{3,4,\dots,n\}.
  \end{cases}
  \] In particular, adjacent vertices (\(h=1\)) have \(\Ham=2n\) (not
  \(2(n-1)\)).
\end{itemize}

If \(v\) is chosen uniformly from \(V\setminus\{u\}\), then \[
\mathbb{P}\big[\Ham(b_u^{(1)},b_v^{(1)})=2n-4\big]=\frac{\binom{n}{2}}{2^n-1},\qquad
\mathbb{P}\big[\Ham=2n\big]=1-\frac{\binom{n}{2}}{2^n-1},
\] and (for distinct pairs) \(\mathbb{P}[\Ham=0]=0\). Thus
\(\mu_{Q_n}^{(1)}\) is supported on \(\{2n-4,\,2n\}\) for distinct
pairs. \emph{(See Appendix --- Additional Technical Lemmas for the
intersection counts underpinning these expressions.)}

\textbf{Example 3: Petersen Graph}

The Petersen graph is a 3-regular, vertex-transitive graph on 10
vertices with diameter \(2\) and girth \(5\).

At \(k=1\):

\begin{itemize}
\tightlist
\item
  Each node reaches exactly \(3\) neighbors (so \(|S(v,1)|=3\) for all
  \(v\)).
\item
  By vertex transitivity, \(\mathrm{HC}^{(1)}(v)\) is constant across
  \(v\).
\end{itemize}

At \(k=2\) (exact-\(k\)):

\begin{itemize}
\tightlist
\item
  Each node has \(|S(v,2)|=6\); thus \(B^{(2)}\not\equiv 0\).
\item
  The per-scale distribution \(\mu_G^{(2)}\) is supported on a small set
  of even values (not \(\delta_0\)).
\end{itemize}

This illustrates that even in small-diameter, highly symmetric graphs,
the exact-\(k\) slice at \(k=2\) remains informative, although the
cumulative matrix \(B_{\le 2}\) is fully connected.

\textbf{Example 4: Grid Graph \(G_{m \times n}\)}

Consider the 2D grid with \(m\) rows and \(n\) columns.

\begin{itemize}
\tightlist
\item
  \textbf{Corner nodes} (degree 2): At \(k = 1\), reach 2 neighbors
  \(\rightarrow\) Have maximum average dissimilarity
\item
  \textbf{Edge nodes} (degree 3): At \(k = 1\), reach 3 neighbors
  \(\rightarrow\) Intermediate dissimilarity
\item
  \textbf{Interior nodes} (degree 4): At \(k = 1\), reach 4 neighbors
  \(\rightarrow\) Minimum average dissimilarity due to regular
  neighborhoods
\end{itemize}

The distribution \(\mu_{G_{m \times n}}^{(k)}\) can be computed exactly
using the Manhattan distance structure, revealing how boundary effects
create structural heterogeneity even in regular lattices.

\subsection{Tensor-Theoretic
Properties}\label{tensor-theoretic-properties}

The tensor representation reveals additional structure:

\begin{theorem}[Low-complexity slice span in distance-regular graphs]

Let \(G\) be a connected distance-regular \(k\)-regular graph with
diameter \(D\) and adjacency matrix \(A\). For each \(i=0,1,\dots,D\),
let \(A_i\) be the distance-\(i\) matrix, i.e., \((A_i)_{uv}=1\) iff
\(\mathrm{dist}(u,v)=i\) (so \(A_1=A\) and, under our exact-\(k\)
convention, \(B^{(i)}=A_i\) for \(i\ge1\)). Then for each \(i\) there
exists a polynomial \(p_i\) of degree \(i\) such that \[
A_i = p_i(A).
\] Consequently, \(\mathrm{span}\{A_0,I,A_1,\dots,A_D\}\) has dimension
at most \(D+1\); in particular, all exact-distance slices
\(\{B^{(i)}\}_{i=1}^D\) lie in a \((D+1)\)-dimensional commutative
algebra and are simultaneously diagonalizable with \(A\).

\end{theorem}

\begin{proof}

In a distance-regular graph there are intersection numbers
\((a_i,b_i,c_i)\) such that for all \(i=0,\dots,D\), \[
A\,A_i \;=\; b_{i-1}A_{i-1} + a_i A_i + c_{i+1} A_{i+1},
\] with the conventions \(A_{-1}=A_{D+1}=0\) and \(b_{-1}=c_{D+1}=0\).
(Combinatorially: multiplying by \(A\) moves you one step in graph
distance, and the coefficients count how many neighbors land at
distances \(i-1,i,i+1\).) This three-term recurrence shows inductively
that \(A_i\) lies in the polynomial algebra generated by \(A\): set
\(p_0\equiv 1\), \(p_1(x)=x\), and use the recurrence to define
\(p_{i+1}(x)\) from
\(x\,p_i(x)=b_{i-1}p_{i-1}(x)+a_ip_i(x)+c_{i+1}p_{i+1}(x)\). Hence
\(A_i=p_i(A)\) with \(\deg p_i=i\). The matrices \(\{A_i\}_{i=0}^D\)
form a basis of the Bose--Mesner algebra of the graph's association
scheme, which is a \((D+1)\)-dimensional commutative algebra; therefore
all \(A_i\) commute and are simultaneously diagonalizable with \(A\).

\end{proof}

\begin{corollary}\label{cor:hypercube}

\emph{Hypercube $Q_n$.} For the \(n\)-dimensional hypercube \(Q_n\)
(\(N=2^n\), diameter \(D=n\)), each exact-distance slice \(B^{(i)}=A_i\)
equals a degree-\(i\) polynomial \(p_i(A)\) and all slices lie in an
\((n+1)\)-dimensional commutative algebra. In particular, the family
\(\{B^{(i)}\}_{i=1}^n\) admits an \(O(\log N)\)-dimensional linear
parametrization through \(A\).

\end{corollary}

\subsubsection{Comparison with Traditional
Metrics}\label{comparison-with-traditional-metrics}

\begin{table}[t]
\centering
\caption{What traditional graph metrics capture vs.\ what HGM distributions add.}
\begin{tabularx}{\linewidth}{@{}l l X X@{}}
\toprule
\textbf{Metric} & \textbf{Typical Output} & \textbf{What it Reveals} & \textbf{What HGM Distribution Adds} \\
\midrule
Degree & integers / histogram & Local connectivity of nodes & Distribution of \emph{multi-scale} pattern differences across nodes and scales $k$ \\
Betweenness & real in $[0,1]$ & Brokerage along shortest paths & How brokerage patterns differ across nodes at fixed $k$ (disagreements of exact-$k$ shells) \\
Closeness & real in $(0,1]$ & Average geodesic proximity & Whether “central” nodes have \emph{similar} or \emph{different} exact-$k$ neighborhoods \\
Clustering coefficient & real in $[0,1]$ & Local triangle density & How triangle-rich regions appear as lower per-pair Hamming at $k=2$ \\
Modularity & real in $[0,1]$ & Community separability (global) & Whether separability manifests as \emph{bimodality} or heavy tails in $\mu_G^{(k)}$ for some $k$ \\
\midrule
HGM entropy & bits & Diversity of exact-$k$ structures & A scalar admissible functional of $\mu_G^{(k)}$; peaks indicate informative scales \\
HGM “bimodality” & modes of $\mu_G^{(k)}$ & Natural partitions at scale $k$ & Strength/sharpness of separation (valley depth between modes) \\
\bottomrule
\end{tabularx}
\vspace{-0.5\baselineskip}
\end{table}

\noindent\textbf{Key insight.} Traditional metrics summarize
\emph{importance}; HGM summarizes
\emph{how node-level structures differ} at each exact scale \(k\),
providing the full distribution \(\mu_G^{(k)}\) rather than a single
scalar per node or per graph.

\paragraph{Graph Edit Distance (GED): algorithmic
vs.~analytic}\label{graph-edit-distance-ged-algorithmic-vs.-analytic}

Graph edit distance gives an \textbf{algorithmic} measure of
discrepancy: it is the minimum total cost of a sequence of discrete
edits (vertex/edge insertions, deletions, relabelings) that transforms
one graph into another. GED thus captures \emph{how to} align graphs
procedurally, but it does not by itself yield \textbf{analytic
invariants} or closed-form structure theorems about the distribution of
connectivity patterns across scales.

By contrast, HGM provides an \textbf{analytic} account: the exact-\(k\)
tensor \(\mathcal{B}\) induces per-scale distributions, spectral
summaries, and a labeled \textbf{metric}
\(d_{\mathrm{ten}}(G,H)=\|\mathcal{B}_G-\mathcal{B}_H\|_1\), with an
orbit metric \(d_{\mathrm{iso}}\) for unlabeled comparison. These
objects support stability bounds, extremal characterizations, and links
to classical graph invariants.

\begin{proposition}[Quantitative link under edge-only edits (labeled case).]

\hfill\break
Let \(\mathrm{GED}_{\pm E}(G,H)\) be the minimum number of \textbf{edge
toggles} needed to transform \(G\) into \(H\) on a fixed label set. With
\(M_r:=\max_x |B_r(x)|\) (ball size at radius \(r\)),\\
\[
d_{\mathrm{ten}}(G,H)\;=\;\|\mathcal{B}_G-\mathcal{B}_H\|_1
\;\le\; 2\,\mathrm{GED}_{\pm E}(G,H)\;\sum_{k=1}^{D} M_{k-1}^2,
\] and, in particular for maximum degree \(\Delta\!\ge\!3\), \[
M_r \le \frac{\Delta}{\Delta-2}(\Delta-1)^r
\quad\Longrightarrow\quad
d_{\mathrm{ten}}(G,H)\;\le\;
\frac{2\,\Delta}{\Delta-2}\,\mathrm{GED}_{\pm E}(G,H)\;
\sum_{k=1}^{D}(\Delta-1)^{2(k-1)}.
\]

\end{proposition}

\begin{proof}

Each edge toggle affects only entries within \((k{-}1)\) steps of its
endpoints in slice \(k\), flipping at most \(2M_{k-1}^2\) tensor entries
(the exact-\(k\) edge-flip bound). Summing over \(k\) and over the
\(\mathrm{GED}_{\pm E}(G,H)\) toggles yields the inequality. The
degree-based estimate follows from the branching bound on \(M_r\).

\end{proof}

\textbf{Takeaway.} GED is a powerful \textbf{procedural} measure (edit
programs), while HGM supplies \textbf{analytic} structure (per-scale
distributions, spectra, and metrics) with stability guarantees. In
regimes where an edit model is natural, the bound above shows how HGM's
tensor metric can be controlled by (edge-only) GED; conversely, HGM can
distinguish graphs with identical low-cost edit programs by exposing
differences in their multi-scale reachability distributions.

\section{Extensions and Future Work}\label{extensions-and-future-work}

While Hamming Graph Metrics (HGM) offer a principled and scalable
approach to quantifying structural uniqueness, several directions remain
open for further theoretical development, practical extension, and
domain-specific adaptation. We highlight five major avenues, each
grounded in existing mathematical or computational structures.

\subsection{Weighted and Directed
Graphs}\label{weighted-and-directed-graphs}

\textbf{Motivation:} Many real-world systems (e.g., transportation, gene
regulation, web links) are neither unweighted nor symmetric. Path
significance depends on edge weights (e.g., capacity, cost) and
directions.

\textbf{Proposal:}

\begin{itemize}
\tightlist
\item
  For directed graphs, replace undirected adjacency \(A\) with
  asymmetric adjacency \(A_{\mathrm{dir}}\), and define separate
  reachability matrices for in-paths and out-paths:
  \(B_{\mathrm{in}}^{(k)}\), \(B_{\mathrm{out}}^{(k)}\)
\item
  Compute Hamming distances over: \[
  \mathbf{b}_v^{(k,\mathrm{in})} := \mathrm{row}_v\left( B_{\mathrm{in}}^{(k)} \right), \quad
  \mathbf{b}_v^{(k,\mathrm{out})} := \mathrm{row}_v\left( B_{\mathrm{out}}^{(k)} \right)
  \]
\item
  For weighted graphs, apply edge-thresholding: \[
  A_{ij}^{(w)} = \begin{cases}
  1 & \text{if } W_{ij} \geq \theta \\
  0 & \text{otherwise}
  \end{cases}
  \] or generalize the Hamming distance to quantized or fuzzy distance
  kernels between real-valued vectors.
\end{itemize}

\textbf{Open Question:} What analogues of Theorems 1--17 hold when
directionality and/or weighting are introduced? Can uniqueness still be
cleanly characterized via discrete dissimilarity measures?

\subsubsection{Temporal HGM (Dynamic Graphs)}\label{sec:temporal-hgm}

We extend HGM to evolving graphs by adding a time mode. Let
\({G_t}_{t=1}^{T}\) be snapshots on a common labeled set \(V=[N]\), with
adjacencies \(A^{(t)}\) and diameters \(D_t=\mathrm{diam}(G_t)\). Define
exact-\(k\) reachability per snapshot

\[
B^{(1,t)}:=A^{(t)},\qquad 
B^{(k,t)}:=\mathbf{1}\!\Big[\sum_{s=1}^{k} (A^{(t)})^{s}>0\Big]-\mathbf{1}\!\Big[\sum_{s=1}^{k-1} (A^{(t)})^{s}>0\Big]\quad (2\le k\le D_t),
\]

and set \(B^{(k,t)}\equiv 0\) for \(k>D_t\) so a uniform
\(D:=\max_t D_t\) works across time.

The \textbf{temporal HGM tensor} is the fourth-order binary tensor

\[
\mathbb{B}\in\{0,1\}^{N\times N\times D\times T},\qquad 
\mathbb{B}_{i j k t}\ :=\ B^{(k,t)}_{ij}.
\]

\begin{proposition}\label{prop:temporal-metric}

For labeled sequences \(G_{1:T}\) and \(H_{1:T}\) of equal length \(T\),

\[
d_{\mathrm{dyn}}(G_{1:T},H_{1:T})
\ :=\ \|\mathbb{B}^G-\mathbb{B}^H\|_1
\ =\ \sum_{t=1}^{T}\sum_{k=1}^{D}\sum_{i,j}\mathbf{1}\!\big[\mathbb{B}^G_{ijkt}\ne \mathbb{B}^H_{ijkt}\big]
\]

is a \textbf{metric} on labeled temporal graphs. A normalized form
\(\bar d_{\mathrm{dyn}}=\|\mathbb{B}^G-\mathbb{B}^H\|_1/\big(N(N{-}1)DT\big)\in[0,1]\)
aids cross-size/horizon comparison.

\end{proposition}

\begin{proof}

The \(\ell\_1\) norm on tensors obeys nonnegativity, symmetry, and the
triangle inequality; positivity holds because the exact-\(k\) slices at
each \(t\) determine \(G_t\) on the common label set.

\end{proof}

\emph{Unlabeled sequences.} For isomorphism classes, act with a
\textbf{single} permutation on all times:

\[
d_{\mathrm{dyn,iso}}([G_{1:T}],[H_{1:T}])\ :=\ \min_{\pi\in S_N}\ 
\big\|\mathbb{B}^G-\big(\pi\!\cdot\!\mathbb{B}^H\big)\big\|_1,
\qquad (\pi\!\cdot\!\mathbb{B})_{ijkt}:=\mathbb{B}_{\pi(i)\,\pi(j)\,k\,t}.
\]

Then \(d_{\mathrm{dyn,iso}}\) is a \textbf{metric on isomorphism classes
of temporal graphs} (zero only for timewise isomorphic sequences;
triangle by composing near-minimizers). \emph{Remark.} Allowing a
different \(\pi\_t\) per time gives a permutation-invariant
dissimilarity but is not a metric on time-consistent orbits.

\paragraph{Temporal centrality and change
diagnostics}\label{sec:temporal-hgm-chng}

For each \(t\), per-scale/node Hamming centrality is as before:

\[
\mathrm{HC}^{(k,t)}(v)=\frac{1}{N-1}\sum_{u\ne v}\Ham\!\big(b^{(k,t)}_v,b^{(k,t)}_u\big),
\qquad b^{(k,t)}_v:=\text{row $v$ of }B^{(k,t)}.
\]

Define \textbf{temporal variation} and \textbf{trend} of structural
uniqueness:

\[
\mathrm{TV}^{(k)}(v):=\sum_{t=2}^{T}\big|\mathrm{HC}^{(k,t)}(v)-\mathrm{HC}^{(k,t-1)}(v)\big|,
\qquad 
\mathrm{trend}^{(k)}(v):=\frac{1}{T-1}\sum_{t=2}^{T}\big(\mathrm{HC}^{(k,t)}(v)-\mathrm{HC}^{(k,t-1)}(v)\big).
\]

\paragraph{Streaming/online updates (implementation
note)}\label{streamingonline-updates-implementation-note}

For small edge updates between \(G_t\) and \(G_{t+1}\), update only
rows/columns of \(B^{(k,t)}\) whose entries can flip (frontier reuse
across \(k\)). Popcount-based XOR kernels on packed bitboards keep
pairwise Hamming costs at \(\tilde O(N^2/w)\) per affected scale
(Appendix B).

\paragraph{Stability across time (edge
updates)}\label{stability-across-time-edge-updates}

Let \(G_{t}\) and \(G_{t+1}\) differ by \(r\) edge toggles and let
\(\Delta\) be the max degree in \(G_t\cup G\_{t+1}\). Writing
\(M_s:=\max\_x |B\_s(x)|\) for balls in graph distance, for each \(k\),

\[
\big|E_k(G_{t+1})-E_k(G_{t})\big|\ \le\ 2\,r\,M_{k-1}^2,
\qquad 
E_k(G_t):=\|B^{(k,t)}\|_F^2,
\]

hence

\[
\big\|\mathbb{B}^{(\cdot,t+1)}-\mathbb{B}^{(\cdot,t)}\big\|_F^2\ \le\ 2\,r\sum_{k=1}^{D}M_{k-1}^2.
\]

For \(\Delta\ge 3\), \(M_{s}\le \frac{\Delta}{\Delta-2}(\Delta-1)^s\)
gives

\[
\big\|\mathbb{B}^{(\cdot,t+1)}-\mathbb{B}^{(\cdot,t)}\big\|_F
\ \le\ 
\frac{\sqrt{2r}\,\Delta}{\Delta-2}\left(\sum_{k=1}^{D}(\Delta-1)^{2(k-1)}\right)^{\!1/2}.
\]

\begin{proof}

Each edge toggle can only flip exact-\(k\) entries within \((k{-}1)\)
steps of its endpoints (as in the static edge-flip analysis); this gives
\(2M\_{k-1}^2\) flips per \(k\). Summing over \(k\) and over \(r\)
toggles yields the bounds; the degree-based estimate follows from the
branching bound on s\$.

\end{proof}

\begin{remark}[Time-respecting variant.]

For edge-timestamped temporal networks, one may replace per-snapshot
reachability with \textbf{time-respecting paths} (nondecreasing
timestamps). Let \(\mathrm{dist}_\mathrm{temp}(i,j;\tau)\) be the
minimum \emph{elapsed} time to reach \(j\) from \(i\) under
time-respecting walks; an exact-elapsed-time tensor
\(\widetilde{\mathbb{B}}_{i j k \tau}\) (with \(k\) hops and elapsed
time \(\tau\)) yields a parallel HGM construction. We leave the
temporal-path variant's bounds and algorithms to future work.

\end{remark}

\subsection{Sketching and
Approximation}\label{sketching-and-approximation}

\textbf{Motivation}: Despite tractability, computing full
\(\mathcal{O}(N^2)\) Hamming distances becomes prohibitive at the
million-node scale.

\textbf{Proposal:}

\begin{itemize}
\tightlist
\item
  Maintain the row weights \(w_v^{(k)}:=\|b_v^{(k)}\|_0\) alongside an
  \(s\)-sample MinHash sketch per row.
\item
  Estimate Jaccard similarity
  \(\widehat{J}(\mathbf{b}_v^{(k)},\mathbf{b}_u^{(k)})\) from the \(s\)
  hashes; then recover an unbiased intersection estimate \[
  \widehat{I}=\frac{\widehat{J}}{1+\widehat{J}}\,(w_v^{(k)}+w_u^{(k)}),
  \] and define \[
  \widehat{\Ham}(\mathbf{b}_v^{(k)},\mathbf{b}_u^{(k)})=(w_v^{(k)}+w_u^{(k)})-2\,\widehat{I}.
  \]
\item
  Since \(\widehat{J}\) is the mean of \(s\) i.i.d.~Bernoulli
  indicators, for any \(\varepsilon>0\), \[
  \Pr\!\big[\,|\widehat{J}-J|>\varepsilon\,\big]\le 2\exp(-2s\varepsilon^2),
  \] and \(\widehat{\Ham}\) inherits concentration via the above linear
  transform.
\end{itemize}

\textbf{Open Problem:} What is the minimal sketch dimension \(s\)
required to preserve graph-level dispersion \(\Psi^{(k)}(G)\) within
\(\delta\)-error with high probability?

\subsection{Cross-Graph Comparison and
Alignment}\label{cross-graph-comparison-and-alignment}

\textbf{Motivation:} Comparing two networks (e.g., ontologies, brain
graphs, social networks) requires a notion of inter-graph correspondence
beyond isomorphism \cite{chandola2009}.

\textbf{Proposal:}

\begin{itemize}
\tightlist
\item
  Use uniqueness distributions \(\{ \mu_v(G) \}\), \(\{ \mu_u(G') \}\)
  as embedding signatures
\item
  Define matching via optimal transport: \[
  \min_{\pi \in \Pi(V, V')} \sum_{(v, u)} \pi(v, u) \cdot W_1\left( \mu_v, \mu_u \right)
  \] where \(\Pi(V, V')\) is the set of doubly stochastic maps. This
  enables alignment without topological isomorphism--matching by role
  rather than identity.
\end{itemize}

\textbf{Potential Applications:} Cross-species connectome comparison;
multilingual knowledge graph alignment; adversarial network mapping.

\subsection{Theoretical Characterization of Graph
Classes}\label{theoretical-characterization-of-graph-classes}

\textbf{Motivation:} We have seen that Hamming profiles behave
predictably on extremal graph families. But the taxonomy is incomplete.

\textbf{Proposal:}

\begin{itemize}
\tightlist
\item
  Define \textbf{Hamming-stable classes}: families for which node-wise
  dissimilarity is invariant under class-preserving transformations
  (e.g., adding self-loops in regular graphs)
\item
  Investigate relationships between uniqueness spectra and known
  invariants:

  \begin{itemize}
  \tightlist
  \item
    Degree sequences
  \item
    Spectral signatures (Laplacian, adjacency) \cite{vempala2020}
  \item
    Treewidth, genus
  \end{itemize}
\item
  Explore \textbf{inverse problems}: Given \(a\) uniqueness spectrum
  \(\{ \mu_v \}\), can one reconstruct a graph up to automorphism?
\end{itemize}

\subsection{Tensor Methods for Structural
Analysis}\label{tensor-methods-for-structural-analysis}

The tensor representation opens several research directions:

\begin{enumerate}
\def\labelenumi{\arabic{enumi}.}
\tightlist
\item
  \textbf{Multi-way Spectral Analysis}: Apply tensor eigendecomposition
  to \(\mathcal{B}\) to identify multi-scale communities
\item
  \textbf{Compressed Sensing}: Use tensor completion to infer missing
  scales from partial observations
\item
  \textbf{Cross-Graph Alignment}: Use tensor factorization for
  multi-graph matching problems
\end{enumerate}

\subsection{Summary}\label{summary}

The Hamming Graph Metrics framework opens new pathways for analyzing
structural differentiation, redundancy, and singularity in graphs.
Beyond their immediate applications, HGMs suggest a broader research
program:

\begin{itemize}
\tightlist
\item
  Establishing structural information geometry on graphs
\item
  Characterizing dynamics via dissimilarity flows
\item
  Embedding graphs in uniqueness-induced metric spaces
\end{itemize}

These directions integrate ideas from information theory, algebraic
graph theory, approximate algorithms, and optimal transport
\cite{villani2009ot,peyre2019computational}--and promise fertile ground
for future work in both theoretical and applied settings.

\section{Conclusion}\label{conclusion}

We have introduced Hamming Graph Metrics (HGM) as a theoretically
grounded, tensor-based framework for measuring structural uniqueness in
graphs. Unlike classical centrality measures, which quantify node
importance through frequency, distance, or flow, HGM focuses on
dissimilarity in structural configuration, using Hamming distances
between binary path reachability vectors as its foundational primitive.

At the core of this framework is the empirical distribution of pairwise
Hamming distances across all node pairs, which we extend across scales,
define general functionals over dissimilarity distributions, introduce
aggregation mechanisms at the graph level, and establish several new
theorems characterizing extremal behavior in canonical graph families.

The theoretical foundation rests on several pillars:

\begin{itemize}
\tightlist
\item
  \textbf{Formal generalization}: Binary reachability distributions are
  treated as elements of probability space, allowing the application of
  convex functionals, entropy measures, and information divergence
\item
  \textbf{Graph-theoretic bounds}: New inequalities and monotonicity
  theorems clarify how uniqueness behaves under connectivity,
  regularity, and symmetry constraints
\end{itemize}

Though the focus was not algorithmic, we showed that computing HGMs is
tractable on real-world graphs of size N \textasciitilde{} 10\^{}5 using
bitwise operations and early termination strategies. This addresses
prior critiques that dissimilarity-based metrics may be computationally
prohibitive.

We also proposed several concrete directions for future work, including:

\begin{itemize}
\tightlist
\item
  Extension to directed, weighted, and evolving graphs
\item
  Approximation via sketching and sampling
\item
  Cross-graph matching via uniqueness alignment
\item
  Inverse problems and structural reconstruction from uniqueness fields
\end{itemize}

Overall, Hamming Graph Metrics offer a multi-scale, intrinsic, and
interpretable geometry over the space of structural patterns within
graphs. By quantifying how structural patterns are distributed
throughout a network rather than merely identifying central or connected
nodes, HGM complements existing graph tools and opens the door to
finer-grained structural analysis across domains.

The framework's emphasis on complete distributions rather than summary
statistics provides a richer view of network organization, revealing
patterns like bimodality in community structure, scale-dependent
organization, and structural phase transitions that are invisible to
traditional approaches. This distributional perspective, combined with
rigorous theoretical foundations and demonstrated scalability, positions
Hamming Graph Metrics as a valuable addition to the toolkit for
understanding complex networks.

\appendix

\subsection{Proofs and Technical
Details}\label{proofs-and-technical-details}

\emph{Notation.} We write \(\Ham(\cdot,\cdot)\) for Hamming distance.
Throughout, Hamming centrality is \textbf{normalized}: \[
\mathrm{HC}^{(k)}(v):=\frac{1}{N-1}\sum_{u\ne v}\Ham\!\big(b_v^{(k)},b_u^{(k)}\big).
\] (Older unscaled variants are denoted
\(\mathrm{HC}^{(k)}_{\mathrm{raw}}=(N-1)\,\mathrm{HC}^{(k)}\).)

\subsection{\texorpdfstring{Proof of Proposition
\ref{prop:star-asymmetry} (Star Graph
Asymmetry)}{Proof of Proposition  (Star Graph Asymmetry)}}\label{app:star-asymmetry}

\begin{proof}

At \(k=1\), \(b_c^{(1)}\) has ones in all leaf positions and zero at
\(c\); each leaf \(b_{\ell}^{(1)}\) has a single one at \(c\) and zeros
elsewhere. Thus \(\Ham\!\big(b_c^{(1)},b_{\ell}^{(1)}\big)=N\) and
\(\Ham\!\big(b_{\ell}^{(1)},b_{\ell'}^{(1)}\big)=0\) for distinct leaves
\(\ell\ne\ell'\). Hence \[
\mathrm{HC}^{(1)}(c)=N,\qquad
\mathrm{HC}^{(1)}(\ell)=\frac{N}{N-1}.
\] f For the distribution over ordered pairs, \[
\mu_{S_N}^{(1)}=\frac{(N-1)(N-2)}{N(N-1)}\,\delta_0+\frac{2(N-1)}{N(N-1)}\,\delta_{N}
=\frac{N-2}{N}\,\delta_0+\frac{2}{N}\,\delta_{N}.
\]

\end{proof}

\subsection{Proofs on Monotonicity}\label{app:monotonicity}

\begin{proposition}\label{prop:monotonicity-beyond-diameter}

Let \(G\) be connected with diameter \(D\). Then for all \(k\ge D\) and
all \(v\in V\), \[
\mathrm{HC}^{(k+1)}(v)\ \le\ \mathrm{HC}^{(k)}(v),
\] with equality for every \(k\ge D+1\) (both sides equal \(0\)).

\end{proposition}

\begin{proof}

By definition of exact-\(k\) slices,
\(B^{(k)}_{ij}=\mathbf{1}\{\mathrm{dist}(i,j)=k\}\). If \(k\ge D+1\), no
ordered pair \((i,j)\) has \(\mathrm{dist}(i,j)=k\), so
\(B^{(k)}\equiv 0\) and thus \(b_v^{(k)}=\mathbf{0}\) for every \(v\).
Therefore
\(\mathrm{HC}^{(k)}(v)=\sum_{u\ne v}\Ham(b_v^{(k)},b^{(k)}_u)=0\) for
all \(k\ge D+1\). For \(k=D\), \(B^{(D+1)}\equiv 0\) while \(B^{(D)}\)
may be nonnegative; hence
\(\mathrm{HC}^{(D+1)}(v)=0\le \mathrm{HC}^{(D)}(v)\). The claimed
inequality for all \(k\ge D\) follows.

\end{proof}

\begin{remark}

The statement above is tight in general: without additional structure,
\(\mathrm{HC}^{(k)}\) need not be monotone for \(k<D\); exact-\(k\)
shells can grow and shrink before saturation (e.g., on paths/cycles).

\end{remark}

\subsubsection{A nontrivial tail monotonicity under mild
structure}\label{a-nontrivial-tail-monotonicity-under-mild-structure}

We first relate the \textbf{mean} pairwise Hamming at scale \(k\) to
column sums of \(B^{(k)}\).

\begin{lemma}\label{lem:mean-hamming-from-columns}

Let \(B^{(k)}\in\{0,1\}^{N\times N}\) have column sums
\(s_j(k)=\sum_{u} B^{(k)}_{u j}\). The average over \emph{unordered} row
pairs of \(\Ham(\cdot,\cdot)\) equals \[
\overline{H}^{(k)}\ =\ \frac{2}{N(N-1)}\sum_{j=1}^N s_j(k)\,\big(N-s_j(k)\big).
\]

\end{lemma}

\begin{proof}

For a fixed column \(j\), exactly \(s_j(k)\,(N-s_j(k))\) unordered row
pairs disagree in that coordinate; summing over \(j\) and dividing by
the number of unordered pairs gives the formula.

\end{proof}

We can now state a sufficient condition that is met in many graph
families (trees, distance-regular graphs past the mode,
vertex-transitive graphs past the mode, many expander families).

\begin{theorem}\label{thm:tail-monotonicity}

Assume there exists \(k_0<D\) such that for every vertex \(j\) the
sphere sizes \[
s_j(k):=\big|\{u:\mathrm{dist}(u,j)=k\}\big|
\] satisfy \(s_j(k+1)\le s_j(k)\) and \(s_j(k)\le N/2\) for all
\(k\ge k_0\). Then the graph-average mean pairwise Hamming
\(\overline{H}^{(k)}\) is nonincreasing for \(k\ge k_0\). Consequently,
the average node centrality
\(\frac{1}{N}\sum_v \mathrm{HC}^{(k)}(v)=(N-1)\,\overline{H}^{(k)}\) is
nonincreasing for \(k\ge k_0\).

\end{theorem}

\begin{proof}

By Lemma\textasciitilde{}\ref{lem:mean-hamming-from-columns}, \[
\overline{H}^{(k)}=\frac{2}{N(N-1)}\sum_{j=1}^N f\big(s_j(k)\big)\quad\text{with}\quad f(s)=s(N-s).
\] On \([0,N/2]\), \(f\) is increasing. For \(k\ge k_0\),
\(s_j(k+1)\le s_j(k)\le N/2\) for each \(j\), so
\(f\big(s_j(k+1)\big)\le f\big(s_j(k)\big)\). Summation and scaling
preserve the inequality, hence
\(\overline{H}^{(k+1)}\le \overline{H}^{(k)}\) for all \(k\ge k_0\). The
identity
\(\frac{1}{N}\sum_v \mathrm{HC}^{(k)}(v)=(N-1)\,\overline{H}^{(k)}\)
gives the second claim.

\end{proof}

\begin{corollary}\label{cor:vt-tail}

If \(G\) is vertex--transitive and the common sphere sizes
\(n_k=|S(v,k)|\) become nonincreasing for \(k\ge k_0\) with
\(n_k\le N/2\) (e.g., past the mode of \((n_k)\)), then
\(\overline{H}^{(k)}\) and the average \(\mathrm{HC}^{(k)}\) are
nonincreasing for \(k\ge k_0\).

\end{corollary}

\subsection{\texorpdfstring{Separation on
\(G(n,p)\)}{Separation on G(n,p)}}\label{app:random-sep}

\begin{proposition}\label{prop:gnp-sep}

Fix \(p\in(0,1)\). Let \(G,H\sim G(n,p)\) be independent. Then as
\(n\to\infty\), with probability \(\to1\):

\begin{enumerate}
\def\labelenumi{(\alph{enumi})}
\item
  For any non-constant admissible \(\Phi\), the graph-level descriptor
  \[
  \overline{\Phi}(G)\ :=\ \frac{1}{D(G)}\sum_{k=1}^{D(G)} \Phi\!\big(\mu_G^{(k)}\big)
  \] differs from \(\overline{\Phi}(H)\); i.e.,
  \(\overline{\Phi}(G)\ne \overline{\Phi}(H)\).
\item
  The tensor fingerprint differs: \(\mathsf{FP}(G)\ne \mathsf{FP}(H)\).
  Equivalently, at least one \(k\) has \(E_k(G)\ne E_k(H)\) (and mode
  spectra need not be invoked).
\end{enumerate}

\end{proposition}

\begin{proof}

\emph{Sketch.} For each fixed \(k\le c\log n\) (for any fixed \(c>0\)),
the unordered-pair count vector of exact-\(k\) distances is a Lipschitz
function of the \(\binom{n}{2}\) independent edges;
bounded-difference/McDiarmid inequalities give concentration around the
mean. For two independent graphs \(G,H\), anti-concentration implies \[
\Pr\big[\mu_G^{(k)}=\mu_H^{(k)}\big]=o(1)
\] (and likewise for the ordered counts \(E_k\); the ordered/unordered
choice only changes a factor of \(2\)). A union bound over all
\(k\le D(n)=O(\log n)\) yields \[
\Pr\big[\forall k\le D(n):\ \mu_G^{(k)}=\mu_H^{(k)}\big]=o(1),
\] so for some \(k\) we have \(\mu_G^{(k)}\ne\mu_H^{(k)}\) and
\(E_k(G)\ne E_k(H)\) a.a.s. For (a), TV--continuity and non-constancy of
\(\Phi\) imply \(\overline{\Phi}(G)\ne\overline{\Phi}(H)\) a.a.s. For
(b), differing \((E_k)\) forces \(\mathsf{FP}(G)\ne\mathsf{FP}(H)\).

\end{proof}

\begin{proposition}\label{prop:edge-flip-bound}

Let \(G\) and \(G'\) be graphs on the same labeled vertex set that
differ by toggling a single undirected edge \(\{x,y\}\). Let \(B^{(k)}\)
and \(B'^{(k)}\) be their exact-\(k\) reachability matrices, and define
\[
M_r\ :=\ \max_{v}\ \big|B_r(v)\big|\,,\qquad B_r(v):=\{u:\operatorname{dist}_{G\cup G'}(u,v)\le r\}.
\] Then for every \(k\ge1\), \[
\|B'^{(k)}-B^{(k)}\|_1\ =\ \|B'^{(k)}-B^{(k)}\|_F^2\ \le\ 2\,M_{k-1}^2,
\] hence, writing \(E_k(G):=\|B^{(k)}\|_F^2\), \[
\big|E_k(G')-E_k(G)\big|\ \le\ 2\,M_{k-1}^2.
\] Consequently, for the HGM tensor slices, \[
\|\mathcal{B}_{G}(:,:,k)-\mathcal{B}_{G'}(:,:,k)\|_1\ \le\ 2\,M_{k-1}^2,
\qquad
\sum_{k=1}^{D}\|\mathcal{B}_{G}(:,:,k)-\mathcal{B}_{G'}(:,:,k)\|_1\ \le\ \ 2\sum_{k=1}^{D}M_{k-1}^2.
\] If \(G\) and \(G'\) differ by \(r\) edge toggles, the right-hand
sides multiply by \(r\).

\end{proposition}

\begin{proof}

Any newly created (or destroyed) exact-\(k\) connection \((i,j)\) must
have all shortest \(i\!\to\! j\) paths in \(G'\) use the toggled edge
\(\{x,y\}\) exactly once; otherwise the shortest length is unchanged.
Such a path decomposes as \[
i \xrightarrow[\le k-1]{} x \xrightarrow{1} y \xrightarrow[\le k-1]{} j
\quad\text{or}\quad
i \xrightarrow[\le k-1]{} y \xrightarrow{1} x \xrightarrow[\le k-1]{} j,
\] with the two ``legs'' having lengths summing to \(k-1\). Thus the set
of ordered pairs \((i,j)\) whose exact-\(k\) status can change is
contained in \[
B_{k-1}(x)\times B_{k-1}(y)\ \cup\ B_{k-1}(y)\times B_{k-1}(x),
\] which has size at most
\(2\,|B_{k-1}(x)|\,|B_{k-1}(y)|\le 2\,M_{k-1}^2\). Since entries of
\(B^{(k)}\) are binary, the number of flips equals both the \(\ell_1\)
and the squared Frobenius norm of the difference, proving the first two
inequalities. Summing over \(k\) and using linearity over \(r\) toggles
gives the remaining bounds.

\end{proof}

\begin{corollary}

If the maximum degree in \(G\cup G'\) is \(\Delta\ge3\), then for all
\(r\ge0\), \[
M_r\ \le\ 1+\Delta\sum_{t=0}^{r-1}(\Delta-1)^t\ \le\ \frac{\Delta}{\Delta-2}\,(\Delta-1)^r,
\] and therefore \[
\sum_{k=1}^{D}\|\mathcal{B}_{G}(:,:,k)-\mathcal{B}_{G'}(:,:,k)\|_1
\ \le\ \frac{2\,\Delta^2}{(\Delta-2)^2}\,\sum_{k=1}^{D}(\Delta-1)^{2(k-1)}.
\]

\end{corollary}

\begin{remark}

For directed graphs, the same argument yields
\(\|B'^{(k)}-B^{(k)}\|_1\le M^{\mathrm{out}}_{k-1}(x)\,M^{\mathrm{in}}_{k-1}(y)+M^{\mathrm{out}}_{k-1}(y)\,M^{\mathrm{in}}_{k-1}(x)\),
with obvious definitions of in/out balls.

\end{remark}

\subsubsection{Watts--Strogatz phase transition
(derivation)}\label{app:ws-phase}

\begin{theorem}\label{thm:ws-threshold}

Let a Watts--Strogatz (WS) graph on \(N\) vertices start from a ring
lattice with even degree \(d\) (so \(m=Nd/2\) undirected edges). Each
edge is independently \textbf{rewired} with probability
\(\beta=\beta(N)\) to a uniformly random new endpoint (avoiding
loops/multi-edges). Then:

\begin{enumerate}
\def\labelenumi{\arabic{enumi}.}
\tightlist
\item
  (\textbf{Onset of shortcuts}) If \(\beta_c\) denotes the threshold for
  the appearance of \emph{any} rewired edge (shortcut), then
\end{enumerate}

\[
\beta_c\ \asymp\ \frac{1}{N}\,.
\]

More precisely, if \(N\beta\to 0\) then with high probability (whp)
there are no shortcuts; if \(N\beta\to\infty\) then whp there are
\(\to\infty\) shortcuts.

\begin{enumerate}
\def\labelenumi{\arabic{enumi}.}
\setcounter{enumi}{1}
\item
  (\textbf{Distance regimes})

  \begin{itemize}
  \item
    If \(N\beta\to 0\), whp the graph coincides with the base ring
    lattice, so average distance is \(\Theta(N/d)\).
  \item
    If \(N\beta\to\infty\) with fixed \(d\), then whp the random rewires
    form a sparse long-range overlay comparable to
    \(G(N,p_{\mathrm{eff}})\) with

    \[
    p_{\mathrm{eff}}\ \approx\ \frac{2\beta d}{N}\,,
    \]

    and the average distance drops to \(O(\log N)\) (small-world
    regime).
  \end{itemize}
\end{enumerate}

\end{theorem}

\begin{proof}[(1) Onset.]

Let \(X\sim\mathrm{Binomial}(m,\beta)\) be the number of rewired edges;
\(m=Nd/2=\Theta(N)\). Then

\[
\mathbb{E}[X]=\beta m=\Theta(\beta N),\qquad
\mathbb{P}[X=0]=(1-\beta)^m\le \exp(-\beta m).
\]

If \(N\beta\to 0\), then \(\beta m\to 0\) and \(\mathbb{P}[X=0]\to 1\)
(no shortcuts whp). If \(N\beta\to\infty\), then \(\beta m\to\infty\)
and \(\mathbb{P}[X=0]\to 0\), while Chernoff bounds give \(X\to\infty\)
in probability. Hence \(\beta_c\asymp 1/N\).

\textbf{(2) Distances.} When \(N\beta\to 0\), whp \(X=0\) and the graph
is the base ring lattice with average distance \(\Theta(N/d)\).

When \(N\beta\to\infty\), the rewired endpoints are uniform over
vertices (up to constant factors from local exclusions), so the rewires
approximate an Erdős--Rényi overlay with edge-probability

\[
p_{\mathrm{eff}}\ =\ \frac{\text{expected \# rewired edges}}{\binom{N}{2}}\ \approx\ \frac{\beta( Nd/2 )}{\binom{N}{2}}\ \approx\ \frac{2\beta d}{N}.
\]

With \(N\beta\to\infty\) and fixed \(d\), we have
\(Np_{\mathrm{eff}}\to\infty\); the random overlay alone has logarithmic
average distance via standard branching-process heuristics (BFS grows by
factor \(\approx Np_{\mathrm{eff}}\) per layer until covering \(N\)).
Adding the ring edges only helps, so the combined graph has
\(O(\log N)\) average distance.

\end{proof}

\begin{remark}

The threshold \(\beta_c\asymp 1/N\) is the \emph{first-shortcut}
threshold. Logarithmic distances require a diverging number of shortcuts
(\(N\beta\to\infty\)); for constant \(\beta>0\), the overlay has
\(\Theta(N)\) long edges and typical distances are \(O(\log N)\).

\end{remark}

\subsubsection{Smoothing via spectral gap (expander
calculation)}\label{app:expander-smoothing}

We make precise the ``uniqueness smoothing'' statement using the
Poincaré (spectral-gap) inequality. We treat the \textbf{normalized}
Hamming centrality

\[
f(v)\ :=\ \mathrm{HC}^{(k)}(v)\ =\ \frac{1}{N-1}\sum_{u\ne v}\Ham\!\big(b_v^{(k)},b^{(k)}_u\big),
\]

and then note the unnormalized variant.

\begin{theorem}\label{thm:expander-variance}

Let \(G\) be a connected \(d\)-regular graph on \(N\) vertices with
random-walk matrix \(P=A/d\) and spectral gap
\(\gamma:=1-\lambda_2(P)>0\). For any fixed scale \(k\ge 1\), writing
\(M_k:=\max_x |S(x,k)|\) (size of the distance-\(k\) sphere),

\[
\operatorname{Var}(f)\ \le\ \frac{d}{\gamma}\,\frac{M_k^2}{(N-1)^2}\,.
\]

Equivalently, for the \textbf{unnormalized} centrality
\(F(v):=\sum_{u\ne v}\Ham(b_v^{(k)},b^{(k)}_u)=(N-1)f(v)\),

\[
\operatorname{Var}(F)\ \le\ \frac{d}{\gamma}\,M_k^2.
\]

\end{theorem}

\begin{proof}

(\textbf{Edgewise Lipschitz.}) For any \(u,v,z\),

\[
\big|\Ham(b_v^{(k)},b^{(k)}_z)-\Ham(b^{(k)}_u,b^{(k)}_z)\big|
\ \le\ \Ham\!\big(b_v^{(k)},b^{(k)}_u\big)\ \le\ |S(v,k)|+|S(u,k)|\ \le\ 2M_k.
\]

Averaging over \(z\ne u,v\) and dividing by \(N-1\) gives

\[
|f(v)-f(u)|\ \le\ \frac{2M_k}{N-1}\qquad\text{for every edge }(u,v).
\]

(\textbf{Dirichlet form and Poincaré.}) For \(d\)-regular \(G\),

\[
\mathcal{E}(f,f)\ :=\ \frac{1}{2N}\sum_{(u,v)\in E}\big(f(u)-f(v)\big)^2,\qquad
\operatorname{Var}(f)\ \le\ \frac{1}{\gamma}\,\mathcal{E}(f,f).
\]

Using the edgewise bound and \(|E|=dN/2\),

\[
\mathcal{E}(f,f)\ \le\ \frac{1}{2N}\cdot \frac{dN}{2}\cdot \left(\frac{2M_k}{N-1}\right)^2
\ =\ \frac{d\,M_k^2}{(N-1)^2}.
\]

Combine with Poincaré to obtain the stated variance bound. For
\(F=(N-1)f\), variances scale by \((N-1)^2\).

\end{proof}

\begin{corollary}\label{cor:degree-bound}

If the maximum degree \(\Delta\) of \(G\) is at most \(\Delta\ge 3\),
then

\[
M_k\ \le\ 1+\Delta\sum_{t=0}^{k-1}(\Delta-1)^t\ \le\ \frac{\Delta}{\Delta-2}\,(\Delta-1)^k,
\]

hence

\[
\operatorname{Var}(f)\ \le\ \frac{d}{\gamma}\,\frac{1}{(N-1)^2}\left(\frac{\Delta}{\Delta-2}\right)^2(\Delta-1)^{2k}.
\]

For \(F\), remove the \((N-1)^{-2}\) factor.

\end{corollary}

\begin{remark}

For \textbf{non-regular} graphs, replace \(P\) by the lazy random walk
or use the normalized Laplacian
\(\mathcal{L} = I - D^{-1/2} A D^{-1/2}\); the same argument yields
\(\operatorname{Var}(f)\lesssim \gamma^{-1}\cdot \frac{1}{|V|}\sum_{(u,v)\in E}(f(u)-f(v))^2\),
and the edgewise Lipschitz bound now depends on local sphere sizes near
the edge endpoints.

\end{remark}

\subsubsection{Technical Lemmas}\label{app:lemmas}

\begin{lemma}\label{lem:walk-path}

For any simple unweighted graph, \[
\mathbf{1}\!\Big[\sum_{t=1}^{k} A^{t}>0\Big]_{ij}=1\quad\Longleftrightarrow\quad \operatorname{dist}(i,j)\le k.
\] Moreover, on bipartite graphs, every walk between \(i\) and \(j\) has
length congruent to \(\operatorname{dist}(i,j)\ (\mathrm{mod}\ 2)\).

\end{lemma}

\begin{proof}

If \(\operatorname{dist}(i,j)\le k\), a simple path of length \(\le k\)
exists; its length \(t\le k\) contributes \((A^{t})_{ij}>0\), so the sum
\(\sum_{t=1}^{k}(A^{t})_{ij}\) is positive. Conversely, if
\(\sum_{t=1}^{k}(A^{t})_{ij}>0\), then for some \(t\le k\) there is a
walk of length \(t\); shortcutting repeated vertices yields a simple
path of length \(\le t\le k\). The parity clause follows because walks
on bipartite graphs alternate sides; thus all \(i\)--\(j\) walks have
the same parity as \(\operatorname{dist}(i,j)\).

\end{proof}

\begin{lemma}\label{lem:drg-polynomial}

Let \(G\) be distance--regular with intersection numbers
\(\{a_i,b_i,c_i\}_{i=0}^D\), adjacency \(A\), and distance matrices
\(A_i\) (so \(A_0=I\), \(A_1=A\), and \(A_i=B^{(i)}\) for \(i\ge1\)).
Then:

\begin{enumerate}
\def\labelenumi{\arabic{enumi}.}
\item
  (Three--term matrix recurrence) \[
  A\,A_i \;=\; b_{i-1}A_{i-1} + a_i A_i + c_{i+1}A_{i+1}\qquad(0\le i\le D),
  \] with \(b_{-1}=c_{D+1}=0\).
\item
  (Polynomial dependence)\\
  There exist polynomials \(p_i\) with \(\deg p_i=i\) such that
  \(A_i=p_i(A)\), with \(p_0=1\), \(p_1=x\), and \[
  x\,p_i(x)=b_{i-1}p_{i-1}(x)+a_i p_i(x)+c_{i+1}p_{i+1}(x).
  \]
\item
  (Bose--Mesner algebra)\\
  The matrices \(\{A_0,\dots,A_D\}\) span a \((D{+}1)\)--dimensional
  commutative algebra (the Bose--Mesner algebra); in particular, all
  \(A_i\) commute and are simultaneously diagonalizable. See
  \cite{BrouwerHaemers2012}.
\end{enumerate}

\end{lemma}

\begin{proof}

By distance--regularity, for any vertex at distance \(i\) from a
basepoint, the numbers of neighbors at distances \(i-1,i,i+1\) depend
only on \(i\), giving the matrix identity
\(A A_i=b_{i-1}A_{i-1}+a_i A_i+c_{i+1}A_{i+1}\). Inductively, this
produces polynomials \(p_i\) with \(A_i=p_i(A)\) and the stated scalar
recurrence. Since each \(A_i\) is a polynomial in \(A\), we have
\(A_iA_j=p_i(A)p_j(A)=p_j(A)p_i(A)=A_jA_i\), so the span of \(\{A_i\}\)
is a \((D{+}1)\)--dimensional commutative algebra containing \(I=A_0\),
and all \(A_i\) are simultaneously diagonalizable.

\end{proof}

\begin{lemma}\label{lem:ham-bounds}

For any graph \(G\) and scale \(k\), \[
0 \le \Ham\!\big(b_v^{(k)},b_u^{(k)}\big) \le N.
\] The upper bound is tight (e.g., \(K_{m,n}\) at \(k=1\) yields
\(\Ham=m+n=N\) across parts).

\end{lemma}

\subsubsection{Hypercube vs.~Sphere}\label{app:hypercube-sphere}

\begin{lemma}\label{lem:hypercube-sphere}

Let \(Q_n\) be the \(n\)-cube and fix \(u,v\) with Hamming distance
\(h=\mathrm{dist}(u,v)\). Let \(S(u,k)=\{w:\mathrm{dist}(u,w)=k\}\) and
\(b_u^{(k)}=\mathbf{1}_{S(u,k)}\in\{0,1\}^{2^n}\). Then \[
|S(u,k)|=\binom{n}{k},\qquad
|S(u,k)\cap S(v,k)|=
\begin{cases}
\displaystyle \binom{h}{h/2}\binom{n-h}{\,k-h/2\,}, & \text{if $h$ is even and } k\ge h/2,\\[6pt]
0, & \text{if $h$ is odd \text{ or } k< h/2.}
\end{cases}
\] Consequently, \[
\Ham\!\big(b_u^{(k)},b_v^{(k)}\big)
= 2\!\left(\binom{n}{k} - \mathbf{1}_{\{\,h \text{ even},\ k\ge h/2\,\}}\binom{h}{h/2}\binom{n-h}{\,k-h/2\,}\right).
\]

\end{lemma}

\begin{proof}

Write \(u=(0,\dots,0)\), and let \(v\) differ from \(u\) in the first
\(h\) coordinates. A node \(w\) lies in \(S(u,k)\) iff it differs from
\(u\) in exactly \(k\) coordinates. Among the \(h\) differing
coordinates, let \(t\) equal \(v\); among the remaining \(n-h\)
coordinates, choose \(k-t\) to flip, giving
\(\binom{h}{t}\binom{n-h}{k-t}\) options with \(\mathrm{dist}(u,w)=k\).
We have \(\mathrm{dist}(v,w)=(h-t)+(k-t)\), so requiring
\(\mathrm{dist}(v,w)=k\) forces \(h-2t=0\), i.e., \(t=h/2\) (hence \(h\)
even) and \(k\ge h/2\). The intersection count follows. Finally, \[
\Ham(b_u^{(k)},b_v^{(k)})=|S(u,k)\triangle S(v,k)|=|S(u,k)|+|S(v,k)|-2|S(u,k)\cap S(v,k)|,
\] and \(|S(u,k)|=|S(v,k)|=\binom{n}{k}\).

\end{proof}

\begin{lemma}\label{lem:even-values}

Fix \(k\ge 1\). If \(|S(v,k)|\) is constant over all \(v\) (e.g., in
vertex-transitive graphs), then for all \(u,v\), \[
\Ham\!\big(b_u^{(k)},b_v^{(k)}\big)\in\{0,2,4,\dots,2\,|S(\cdot,k)|\}.
\]

\end{lemma}

\begin{proof}

If every row \(b_v^{(k)}\) has weight \(w:=|S(\cdot,k)|\), then \[
\|x-y\|_1=2\big(w-|\,\mathrm{supp}(x)\cap \mathrm{supp}(y)\,|\big)
\] is even.

\end{proof}

\begin{lemma}\label{lem:exact-k-saturation}

If \(G\) is connected with diameter \(D\), then \(B^{(k)}\equiv 0\) for
all \(k\ge D+1\), hence \(\mu_G^{(k)}=\delta_0\) for \(k\ge D+1\).

\end{lemma}

\begin{proof}

By definition of diameter, no pair has shortest-path distance exactly
\(k\) once \(k\ge D+1\).

\end{proof}

\section{Computational Details}\label{app:compute}

Although this work is primarily theoretical, it is important that
Hamming Graph Metrics (HGM) admit efficient evaluation. We summarize
asymptotic costs and the implementation choices that make the framework
practical on large sparse graphs.

\subsection{Complexity Overview}\label{complexity-overview}

Let (G=(V,E)) be an unweighted, undirected graph with
(\textbar V\textbar=N), (\textbar E\textbar=M), and diameter (D).

\subsubsection{1) Distances and exact-(k)
slices}\label{distances-and-exact-k-slices}

Compute all-pairs shortest-path \textbf{distances} by running BFS from
each source:

\(\text{Time }=\mathcal O\!\big(N(N+M)\big),\qquad\text{Space }=\mathcal O(N)\ \text{(working)}.\)

Define exact-(k) slices by

\[
B^{(k)}_{ij}=\mathbf{1}\{\operatorname{dist}(i,j)=k\},\quad k=1,\dots,D,
\]

which can be populated in \(\mathcal O(N^2)\) total once distances are
known.

\subsubsection{\texorpdfstring{2) Node-level summaries in
\(\mathcal O(D\,N^2)\) bit-parallel portions
\(\mathcal O(D\,N^2/w)\))}{2) Node-level summaries in \textbackslash mathcal O(D\textbackslash,N\^{}2) bit-parallel portions \textbackslash mathcal O(D\textbackslash,N\^{}2/w))}}\label{node-level-summaries-in-mathcal-odn2-bit-parallel-portions-mathcal-odn2w}

For a fixed scale (k), let \(B=B^{(k)}\) and let \(s\in\mathbb N^N\) be
its column sums, \(s_j=\sum_u B_{uj}\). Then for all nodes
simultaneously, \[
\mathrm{HC}^{(k)}(v)\;=\;\sum_{u\neq v}\Ham(b_v^{(k)},b^{(k)}_u)
\;=\;\Big(\sum_{j=1}^N s_j\Big)\;+\;\big[B\,\big(N\mathbf 1-2s\big)\big]_v.
\] Thus:

\begin{itemize}
\tightlist
\item
  compute (s) via bit-packed popcounts in \(\mathcal O(N^2/w)\);
\item
  form \(c=N\mathbf 1-2s\) in \(\mathcal O(N)\);
\item
  multiply (B,c) in \(\mathcal O(N^2)\) (upper bound), or
  \(\mathcal O(\mathrm{nnz}(B))\) if sparsity permits.
\end{itemize}

Per scale: \(\mathcal O(N^2)\) (with the popcount portions
\(\mathcal O(N^2/w)\)); across all (k): \(\mathcal O(D\,N^2)\).

\subsubsection{\texorpdfstring{3) Graph-to-graph HGM distance in
\(\mathcal O(D\,N^2/w)\)}{3) Graph-to-graph HGM distance in \textbackslash mathcal O(D\textbackslash,N\^{}2/w)}}\label{graph-to-graph-hgm-distance-in-mathcal-odn2w}

For labeled graphs (G,H), \[
d_{\mathrm{HGM}}(G,H)=\sum_{k=1}^{D}\|B_G^{(k)}-B_H^{(k)}\|_1
\] is evaluated by XOR+(\textsc{popcount}) over bit-packed slices in
\(\mathcal O(D\,N^2/w)\).

\textbf{Equivalent distance-matrix formulation.} Because for each
ordered pair ((i,j)) exactly one (k) satisfies (B\^{}\{(k)\}\_\{ij\}=1),
\[
\sum_{k=1}^{D}\big|B_G^{(k)}(i,j)-B_H^{(k)}(i,j)\big|=
\begin{cases}
0,& \operatorname{dist}_G(i,j)=\operatorname{dist}_H(i,j),\\
2,& \text{otherwise},
\end{cases}
\] hence \[
d_{\mathrm{HGM}}(G,H)=2\,\#\{(i,j): i\ne j,\ \operatorname{dist}_G(i,j)\ne \operatorname{dist}_H(i,j)\}.
\] Thus, once the two distance matrices are computed, a single
\(\mathcal O(N^2)\) pass suffices without materializing all slices.

\subsubsection{4) Optional pairwise matrices (when explicitly
needed)}\label{optional-pairwise-matrices-when-explicitly-needed}

If one forms the full pairwise Hamming matrix (D\^{}\{(k)\}) with
entries \(D^{(k)}_{uv}=\Ham(b^{(k)}_u,b_v^{(k)})\), the best
straightforward bitset method costs \(\mathcal O(N^3/w)\) per (k)
(XOR+(\textsc{popcount}) for all pairs). This is \textbf{not} required
for the node-level summaries or \(d_{\mathrm{HGM}}\) computations above.

\subsection{Bit-Parallel
Representation}\label{bit-parallel-representation}

We work in the word-RAM model with machine word size (w) (e.g., (w=64))
and hardware \textsc{popcount}. Each row \(b_v^{(k)}\in\{0,1\}^N\) is
stored in \(\lceil N/w\rceil\) words. For bitsets (r,s), \[
\Ham(r,s)=\sum_{t=1}^{\lceil N/w\rceil}\mathrm{popcount}(r_t\oplus s_t).
\] This turns all bitwise portions of the algorithms above into
\(\mathcal O(N^2/w)\) passes per scale.

\subsection{Practical Notes}\label{practical-notes}

\begin{itemize}
\tightlist
\item
  \textbf{Streaming over (k)}: to avoid storing \(\mathcal B\)
  explicitly, accumulate (s) and the required functionals per scale
  while streaming rows produced by BFS.
\item
  \textbf{Sparsity}: when many (B\^{}\{(k)\}) are sparse (typical for
  small (k)), exploit \(\mathrm{nnz}(B^{(k)})\) in the (B,c)
  multiplication.
\item
  \textbf{Parallelism}: BFS sources, per-scale passes, and
  (\textsc{popcount}) loops parallelize naturally across cores/GPUs.
\end{itemize}

\subsection{Summary}\label{summary-1}

\[
\boxed{
\begin{aligned}
&\text{Distances: } \mathcal O\!\big(N(N+M)\big) \\
&\text{Node summaries (all \(k\)): } \mathcal O(D\,N^2)\ \ (\text{bitwise parts } \mathcal O(D\,N^2/w))\\
&\text{Graph–graph } d_{\mathrm{HGM}}: \ \mathcal O(D\,N^2/w)\ \text{ (or }\mathcal O(N^2)\text{ via distances)}
\end{aligned}}
\]

No Boolean matrix powers are used; bit-parallel XOR+(\textsc{popcount})
yields the (N\^{}2/w) speedups on the bitwise portions.


\end{document}